\documentclass{CSML}

\def\dOi{11(3:20)2015}
\lmcsheading%
{\dOi}
{1--24}
{}
{}
{Jul.~\phantom01, 2014}
{Sep.~22, 2015}
{}

\ACMCCS{[{\bf Theory of computation}]: Models of Computation; Formal
  languages and automata theory---Automata over infinite objects} 

\usepackage{hyperref}
\usepackage{gastex}
\usepackage{color}
\usepackage[c]{esvect} % arrows
\theoremstyle{plain}\newtheorem{theorem}[thm]{Theorem}
\theoremstyle{plain}\newtheorem{corollary}[thm]{Corollary}
\theoremstyle{plain}\newtheorem{example}[thm]{Example}
\theoremstyle{plain}\newtheorem{lemma}[thm]{Lemma}
\theoremstyle{plain}\newtheorem{remark}[thm]{Remark}
\def\eg{{\em e.g.}}
\def\cf{{\em cf.}}

\newcommand{\problemx}[3]{
\vspace{3mm}
\par\noindent{\bf#1}\par\nobreak\vskip.2\baselineskip
\begingroup\clubpenalty10000\widowpenalty10000
\setbox0\hbox{\bf INPUT: }\setbox1\hbox{\bf QUESTION: }
\dimen0=\wd0\ifnum\wd1>\dimen0\dimen0=\wd1\fi
\vskip-\parskip\noindent
\hbox to\dimen0{\box0\hfil}\hangindent\dimen0\hangafter1\ignorespaces#2\par
\vskip-\parskip\noindent
\hbox to\dimen0{\box1\hfil}\hangindent\dimen0\hangafter1\ignorespaces#3\par
\endgroup\vspace{3mm}
}

\newcommand{\Z}{\mathbb{Z}}

\newcommand{\N}{\mathbb{N}}
\newcommand{\RP}{\mathbb{R}_{\geq 0}}

\newcommand{\B}{\mathcal{B}}
\newcommand{\mtl}{MTL }

\newcommand{\reg}{\mathsf{reg}}
\newcommand{\locs}{\mathcal{L}}
\newcommand{\edges}{E}

\newcommand{\clocks}{\mathcal{X}}
\newcommand{\clock}{x}

\newcommand{\loc}{\mathit{l}}
\newcommand{\cmax}{\mathsf{cmax}}

\newcommand{\TW}{T\Sigma^+}

\usepackage{todonotes}
\newcommand{\true}{\mathtt{true}}

\newcommand{\U}{\mathsf{U}}

\newcommand{\sep}{\text{ }|\text{ }} % grammar separator

\renewcommand{\succ}{\mathsf{Succ}}

\renewcommand{\frac}{\mathsf{frac}}
\renewcommand{\int}{\mathsf{int}}

\newcommand*\ie{\textit{i.e.}}

\newcommand{\enc}{\mathsf{enc}} % encoding

\newcommand{\A}{\mathcal{A}}
\newcommand{\tcm}{\mathcal{M}}

\newcommand{\cm}{\mathcal{C}}

\newcommand{\Op}{\mathsf{Op}}
\newcommand{\op}{\mathsf{op}}

\newcommand{\gs}{\mathcal{G}} % the set of global states or configs
\newcommand{\hs}{\mathcal{H}} % the set of global states or configs

\makeatletter
\newcommand*{\defeq}{\mathrel{\rlap{%
                     \raisebox{0.3ex}{$\m@th\cdot$}}%
                     \raisebox{-0.3ex}{$\m@th\cdot$}}%
                     =}
\makeatother

\makeatletter
\newcommand*{\ndefeq}{\mathrel{\rlap{%
                     \raisebox{0.3ex}{$\m@th\cdot$}}%
                     \raisebox{-0.3ex}{$\m@th\cdot$}%
                     \rlap{%
                     \raisebox{0.3ex}{$\m@th\cdot$}}
                     \raisebox{-0.3ex}{$\m@th\cdot$}}
                     =}
\makeatother

\newcommand{\transA}{\Rightarrow_{\hspace{-.8mm}\A}}
\newcommand{\transAB}{\Rightarrow_{\hspace{-.8mm}\A\B}}

\begin{document}

	\title[Timed Automata with Unbounded Discrete Data Structures]{Verification for Timed Automata extended with Unbounded Discrete Data Structures}

\author[K.~Quaas]{Karin Quaas}	%required
\address{Universit\"at Leipzig, Germany}	%required
\email{quaas@informatik.uni-leipzig.de}
\thanks{The author is supported by DFG, project QU 316/1-1.}	%optional

%\author[]{Author 2}	%optional
%\address{address2; addresses should be duplicated when authors share an affiliation}	%optional
%\email{author2@email2; ditto for email addresses}  %optional
%\thanks{thanks 2, optional.}	%optional

%\author[]{Author 3}	%optional
%\address{address 3}	%optional
%\email{author3@email3}  %optional
%\thanks{thanks 3, optional.}	%optional

%% etc.

%% required for running head on odd and even pages, use suitable
%% abbreviations in case of long titles and many authors:

%% mandatory lists of keywords and classifications:

\titlecomment{{\lsuper*}An extended abstract was published in \emph{CONCUR 2014}.}
%%%%%%%%%%%%%%%%%%%%%%%%%%%%%%%%%%%%%%%%%%%%%%%%%%%%%%%%%%%%%%%%%%%%%%%%%%%

%% the abstract has to PRECEED the command \maketitle:
%% be sure not to issue the \maketitle command twice!

\begin{abstract}
  \noindent We study decidability of verification problems for timed automata extended with unbounded discrete data structures. 
	More detailed, we extend timed automata with a pushdown stack. 
	In this way, we obtain a strong model that may for instance be used to model real-time programs with procedure calls. 
	It is long known that the reachability problem for this model is decidable. 
	The goal of this paper is to identify subclasses of timed pushdown automata for which the language inclusion problem and related problems are decidable.
\end{abstract}

\maketitle

\section{Introduction}
\emph{Timed automata} were introduced by Alur and Dill~\cite{AD94}, and have since then become a popular standard formalism to model real-time systems. 
An undeniable reason for the success of timed automata is the $\mathsf{PSPACE}$ decidability of the \emph{language emptiness problem}~\cite{AD94}.
A major drawback of timed automata is the undecidability~\cite{AD94} of the \emph{language inclusion problem}: given two timed automata $\A$ and $\B$, does $L(\A)\subseteq L(\B)$ hold?
%given two timed automata $\A$ and $\B$, does $L(A)\subseteq L(\B)$ hold? 
The undecidability of this problem prohibits the usage of automated verification algorithms for analysing timed automata, where $\B$ can be seen as the specification that is supposed to be satisfied by the system modelled by $\A$. 
However, if $\B$ is restricted to have at most one clock, then the language inclusion problem over finite timed words is decidable (albeit with non-primitive recursive complexity)~\cite{DBLP:conf/lics/OuaknineW04}.
As other important milestones in the success story of timed automata we would like to mention the decidability of bisimulation~\cite{DBLP:conf/cav/Cerans92}, 
the decidability of the \emph{model checking problem} for timed automata and a timed extension of CTL~\cite{DBLP:conf/lics/HenzingerNSY92}, and, more recently,  the decidability of the model checking problem for timed automata and Metric Temporal Logic (MTL, for short) over finite timed words~\cite{DBLP:conf/lics/OuaknineW05}.

Timed automata can express many interesting time-related properties, and even with the restriction to a single clock, they  allow one to model a large class of systems, including, for example, the internet protocol TCP~\cite{TCP}.  
If we want to reason about real-time programs with procedure calls, or about the number of events occurring in computations of real-time systems, we have to extend the model of timed automata with unbounded discrete data structures. In 1994, Bouajjani et al.~\cite{DBLP:conf/hybrid/BouajjaniER94} extended timed automata with discrete counters and a pushdown stack, and proved that the satisfiability of reachability properties for several subclasses of this model is decidable. 
Nine years later, it was shown that the binary reachability relation for \emph{timed pushdown systems} is decidable~\cite{DBLP:journals/tcs/Dang03}.  Decidability of the reachability problem was also proved for several classes of \emph{timed counter systems}~\cite{DBLP:journals/entcs/BouchyFS09}, mainly by simple extensions of the classical region-graph construction~\cite{AD94}. 
The language inclusion problem, however, is to the best of our knowledge only considered in~\cite{EmmiM06} for the class of timed pushdown systems. 
In~\cite{EmmiM06} it is stated that the language inclusion problem is decidable if $\A$ is a timed pushdown automaton, and $\B$ is a one-clock timed automaton. 
The proof is based on an extension of the proof for the decidability of the language inclusion problem for the case that $\A$ is a timed automaton without pushdown stack~\cite{DBLP:conf/lics/OuaknineW04}. 
Unfortunately, and as is well known, the proof in~\cite{EmmiM06} is not correct. 

In this paper, 
we prove that different to what is claimed in~\cite{EmmiM06}, the language inclusion problem for the case that $\A$ is a  pushdown timed automaton and $\B$ is a one-clock timed automaton is undecidable. This is even the case if $\A$ is a deterministic instance of a very restricted subclass of timed pushdown automata called \emph{timed visibly one-counter nets}.
On the other hand,
we prove that the language inclusion problem is decidable if $\A$ is a timed automaton and $\B$ is a timed automaton extended with a finite set of counters that can be incremented and decremented, and which we call \emph{timed counter nets}. As a special case, we obtain the decidability of the \emph{universality problem} for timed counter nets: given a timed automaton $\B$ with input alphabet $\Sigma$, does $L(\B)$ accept the set of all timed words over $\Sigma$?
Finally, we give the precise decidability border for the universality problem by proving that the universality problem is undecidable for the class of \emph{timed visibly one-counter automata}. 
We remark that all results apply to extensions of timed automata over \emph{finite} timed words.

\section{Extensions of Timed Automata with Discrete Data Structure}
We use $\Z$, $\N$ and  $\RP$ to denote the integers, the non-negative integers and  the non-negative reals, respectively.

We use $\Sigma$ to denote a finite alphabet. A {\em timed word} over $\Sigma$ is a non-empty finite sequence $(a_1,t_1)\dots(a_k,t_n)\in(\Sigma\times\RP)^+$ such that the sequence $t_1,\dots,t_n$ of timestamps is non-decreasing.
We say that a timed word is \emph{strictly monotonic} if $t_{i-1}<t_{i}$ for every $i\in\{2,\dots,n\}$. 
We use $\TW$ to denote the set of finite timed words over $\Sigma$.
A set $L\subseteq \TW$ is called a {\em timed language}.
%We use $|w|$ to denote the \emph{length} of a timed word $w$. 

Let $\clocks$ be a finite set of {\em clock variables} ranging over $\RP$.
We define {\em  clock constraints} $\phi$ over $\clocks$ to be conjunctions of formulas of the form $\clock\sim c$, where $\clock\in \clocks$, $c\in\N$, and $\sim\in\{<,\leq,=,\geq,>\}$. 
We may use $\true$ as abbreviation for $x\le c\vee x>c$. 
We use $\Phi(\clocks)$ to denote the set of all clock constraints over $\clocks$.
For the case that $\clocks$ is the empty set, we set $\Phi(\clocks)=\{\true\}$. 
A \emph{clock valuation} is a mapping from $\clocks$ to $\RP$. 
A clock valuation $\nu$ satisfies a clock constraint $\phi$, written $\nu\models\phi$, if $\phi$ evaluates to true according to the values given by $\nu$.
For $\delta\in\RP$ and $\lambda\subseteq \clocks$, we define $\nu+\delta$ to be $(\nu+\delta)(\clock)=\nu(\clock)+\delta$ for each $\clock\in \clocks$, and we define $\nu[\lambda\defeq 0]$ by $(\nu[\lambda\defeq 0])(x)=0$ if $x\in\lambda$, and $(\nu[\lambda\defeq 0])(x)=\nu(x)$ otherwise.

Let $\Gamma$ be 
a finite stack alphabet. We use $\Gamma^*$ to denote the set of finite words over $\Gamma$, including the empty word denoted by $\varepsilon$. 
We define a finite set $\Op(\Gamma)$ of \emph{stack operations} by $\Op(\Gamma)\defeq\{pop(a),push(a)\mid a\in\Gamma\}\cup\{noop,empty?\}$.

A \emph{timed pushdown automaton} is a tuple
$\A=(\Sigma, \Gamma, \locs, \locs_0, \locs_f, \clocks, \edges)$, where
\begin{itemize}
\item $\locs$ is a finite set of \emph{locations},
\item $\locs_0\subseteq\locs$ is the set of initial locations,
\item $\locs_f\subseteq \locs$ is the set of accepting locations,
\item $\edges \subseteq \locs\times\Sigma\times\Phi(\clocks)\times\Op(\Gamma) \times 2^\clocks\times\locs$ is a finite set of edges. 
\end{itemize}	
A \emph{state} of $\A$ is a triple $(\loc,\nu,u)$, where $\loc\in\locs$ is the current location, the clock valuation $\nu$ represents the current values of the clocks, and $u\in\Gamma^*$ represents the current stack content, where the top-most symbol of the stack is the left-most 
%\textcolor{red}{left or right?}
symbol in the word $u$, and the empty word $\varepsilon$ represents the empty stack. 
We use $\gs^\A$ to denote the set of all states of $\A$.
A timed pushdown automaton $\A$ induces a  transition relation $\transA$ on $\gs^\A\times\RP\times\Sigma\times\gs^\A$ as follows:
$\langle (\loc,\nu,u),\delta,a,(\loc',\nu',u')\rangle\hspace{-1.6mm}\in\transA$, if, and only if, there exists some edge $(\loc,a,\phi,\op,\lambda,\loc')\in\edges$ 
such that 
$(\nu+\delta)\models\phi$,
$\nu'=(\nu+\delta)[\lambda\defeq 0]$, 
and
(i) if $\op=pop(\gamma)$ for some $\gamma\in\Gamma$, 
then $u=\gamma\cdot u'$;
(ii) if $\op=push(\gamma)$ for some $\gamma\in\Gamma$,
then  $u'=\gamma\cdot u$;
(iii) if $\op=empty?$, then $u=u'=\varepsilon$;
(iv) if $\op=noop$, then $u'=u$. 
A \emph{run} of $\A$ is a finite sequence 
$\prod_{1\leq i\leq n} \langle (\loc_{i-1},\nu_{i-1},u_{i-1}),\delta_i,a_i,(\loc_i,\nu_i,u_i)\rangle$ such that
$\langle(\loc_{i-1},\nu_{i-1},u_{i-1}),\delta_i,a_i,(\loc_i,\nu_i,u_i)\rangle\hspace{-1.2mm}\in\transA$ 
for every $i\in\{1,\dots,n\}$. 
A run is called \emph{successful} if $\loc_0\in\locs_0$, $\nu_0(x)=0$ for every $x\in\clocks$, $u_0=\varepsilon$, and $\loc_n\in\locs_f$.
With a run we associate the timed word $(a_1,\delta_1)(a_2,\delta_1+\delta_2)\dots(a_n,\Sigma_{1\leq i\leq n}\delta_i)$. 
The language accepted by the timed pushdown automaton $\A$, denoted by $L(\A)$, is defined to be the set of timed words $w\in\TW$ for which there exists a successful run of $\A$ that $w$ is associated with.

Next we define some subclasses of timed pushdown automata; see Figure 1 for a graphical overview.
We start with timed extensions of \emph{one-counter automata}~\cite{DLS-tcs10,DBLP:journals/iandc/JancarKMS04} and \emph{one-counter nets}~\cite{DBLP:conf/fsttcs/HofmanLMT13,DBLP:conf/concur/AbdullaC98}. 
A \emph{timed one-counter automaton} is a timed pushdown automaton where the stack alphabet is a singleton. 
By writing $push$ and $pop$ we mean that we increment and decrement the counter, respectively, whereas $empty?$ corresponds to a zero test.  
A \emph{timed one-counter net} is a timed one-counter automaton without zero tests, \ie, the $empty?$ operation is not allowed.
We remark that for both classes, the execution of an edge of the form $(\loc,a,\phi,pop,\lambda,\loc')$ is \emph{blocked} if the stack is empty.

Next, we consider the timed extension of an interesting subclass of pushdown automata called \emph{visibly pushdown automata}~\cite{DBLP:conf/stoc/AlurM04}. 
A \emph{timed visibly pushdown automaton}
is a timed pushdown automaton for which the input alphabet $\Sigma$ can be partitioned into three pairwise disjoint sets $\Sigma=\Sigma_{\mathsf{int}}\cup\Sigma_{\mathsf{call}}\cup\Sigma_{\mathsf{ret}}$ of \emph{internal}, \emph{call}, and \emph{return} input symbols, respectively, and such that
for every edge $(\loc,a,\phi,\op,\lambda,\loc')$
the following conditions are satisfied:
\begin{itemize}
\item $a\in\Sigma_{\mathsf{int}}$ if, and only if,  $\op=noop$,
\item $a\in\Sigma_{\mathsf{call}}$, if, and only if,  $\op=push(b)$ for some $b\in\Gamma$,
\item $a\in\Sigma_{\mathsf{ret}}$ if, and only if,  $\op=empty?$ or $\op=pop(b)$ for some $b\in\Gamma$. 	
\end{itemize}
A \emph{timed visibly one-counter automaton} (timed visibly one-counter net, respectively) is a timed one-counter automaton (timed one-counter net, respectively) that is also a timed visibly pushdown automaton.
We say that a timed visibly one-counter net with no clocks is \emph{deterministic} if 
for all $e=(\loc,a,\true,\op',\emptyset,\loc')$, $e'=(\loc,a,\true,\op'',\emptyset,\loc'')\in \edges$  with $e\neq e'$ we  have either $\op'=pop$ and $\op''=empty?$, or
$\op'=empty?$ and $\op''=pop$. 

Finally, we define the class of \emph{timed counter nets}, which generalizes timed one-counter nets, but is not a subclass of timed pushdown automata. 
A timed counter net of dimension $n$ is a tuple 
$\A=(\Sigma,n,\locs,\locs_0,\locs_f,\clocks, \edges)$, where 
$\locs,\locs_0,\locs_f$ are the sets of locations, initial locations and accepting locations, respectively, and 
$\edges\subseteq \locs\times\Sigma\times\Phi(\clocks)\times\{0,1,-1\}^n\times 2^\clocks\times\locs$ is a finite set of edges. 
A state of a timed counter net is a triple $(\loc,\nu,\vec{v})$, where $\loc\in\locs$, $\nu$ is a clock valuation, and $\vec{v}\in\N^n$ is a vector representing the current values of the counters. We define $\langle(\loc,\nu,\vec{v}),\delta,a,(\loc',\nu',\vec{v'})\rangle\hspace{-1.6mm}\in\transA$ if, and only if, there exists some edge $(\loc,a,\phi,\vec{c},\lambda,\loc')\in\edges$ such that $(\nu+\delta)\models\phi$, $\nu'=(\nu+\delta)[\lambda\defeq 0]$, and $\vec{v'}=\vec{v}+\vec{c}$, where  vector addition is defined pointwise.  Note that, similar to pop operations on an empty stack, transitions which result in the negative value of one of the counters are \emph{blocked}. 
The notions of \emph{runs}, \emph{successful runs}, \emph{associated timed words} and \emph{the language accepted by $\A$}, are defined analogously to the corresponding definitions for timed pushdown automata.

\begin{figure}
\begin{center}
\begin{picture}(123,52)(0,-52)
%\put(0,-52){\framebox(129,52){}}

\node[Nw=37.0,Nh=8.5,Nmr=0.0](n0)(68.0,-9.0){}
\put(52,-8){One-clock timed }
\put(51,-12){pushdown automata}

\node[Nw=41.0,Nh=8.5,Nmr=0.0](n1)(18.0,-9){}

\put(-2,-8){One-clock timed visibly}
\put(1,-12){pushdown automata}

\node[Nw=41.0,Nh=8.5,Nmr=0.0,linewidth=0.5](n2)(18.0,-26.0){}
\put(-2,-25){One-clock timed visibly}
\put(-.5,-29){one-counter automata}

\node[Nw=37.0,Nh=8.5,Nmr=0.0](n3)(68.0,-26.0){}
\put(53,-25){One-clock timed}
\put(49.5,-29){one-counter automata}

%\node[Nw=41.0,Nh=8.5,Nmr=0.0,fillgray=0.9,linewidth=0.5](n4)(18.0,-43.0){}
\node[Nw=41.0,Nh=8.5,Nmr=0.0,linewidth=0.5](n4)(18.0,-43.0){}
\put(-2,-42){One-clock timed visibly}
\put(3,-46){one-counter nets}

%\node[Nw=37.0,Nh=8.5,Nmr=0.0,fillgray=0.9,linewidth=0.5](n5)(68.0,-43.0){}
\node[Nw=37.0,Nh=8.5,Nmr=0.0,linewidth=0.5](n5)(68.0,-43.0){}
\put(53,-42){One-clock timed}
\put(53,-46){one-counter nets}

%\node[Nw=31.0,Nh=8.5,Nmr=0.0,fillgray=0.9,linewidth=0.5](n6)(113.0,-43.0){}
\node[Nw=31.0,Nh=8.5,Nmr=0.0,linewidth=0.5](n6)(113.0,-43.0){}
\put(100,-42){One-clock timed}
\put(103,-46){counter nets}

\node[linecolor=green,Nw=140,Nh=12,Nmr=0.0,linewidth=0.4](nets)(61,-43){}
\node[linecolor=yellow,Nw=96,Nh=32,Nmr=0.0,linewidth=0.4](onecounter)(41,-35){}
\node[linecolor=red,Nw=46,Nh=51,Nmr=0.0,linewidth=0.4](visibly)(18,-28){}

%\node[linecolor=green,Nw=3.5,Nh=4.0,Nmr=0.0](nets)(100,-12){}
%\node[linecolor=blue,Nw=3.5,Nh=4.0,Nmr=0.0](nets)(100,-17){}
%\node[linecolor=red,Nw=3.5,Nh=4.0,Nmr=0.0](nets)(100,-22){}

%\put(110,-8){\footnotesize{A subclass of B}}
%\put(107,-12){\footnotesize{visibly}}
%\put(107,-15){\footnotesize{one counter}}
%\put(107,-20){\footnotesize{no zero test}}

%\node[Nw=3.5,Nh=4.0,Nmr=0.0](o)(100,-7){A}
%\node[Nw=3.5,Nh=4.0,Nmr=0.0](p)(107,-7){B}
%\drawedge(o,p){ }

\drawedge(n4,n2){ }

\drawedge(n2,n1){ }

\drawedge(n1,n0){ }

\drawedge(n3,n0){ }

\drawedge(n5,n3){ }

\drawedge(n4,n5){ }

\drawedge(n2,n3){ }

\drawedge(n5,n6){ }

% for the caption...
\node[Nw=3.5,Nh=4.0,Nmr=0.0](o)(4,-63.5){A}
\node[Nw=3.5,Nh=4.0,Nmr=0.0](p)(11,-63.5){B}
\drawedge(o,p){ }
\end{picture}

\caption{Extensions of one-clock timed automata with discrete data structures.
	Here, 
	\protect\phantom{$A\to B$} means that $A$ is a subclass of $B$.	The classes surrounded by the red frame are visibly pushdown automata, in which the input determines the stack operations. The yellow framed classes only allow one counter, and the green framed classes do not allow zero tests. 
The language emptiness problem is decidable for all classes,
but only the green framed classes have a decidable universality problem. The corresponding results for classes in boxes with bold line are new and presented in this paper.}
\label{figure_A}
\end{center}
\end{figure}
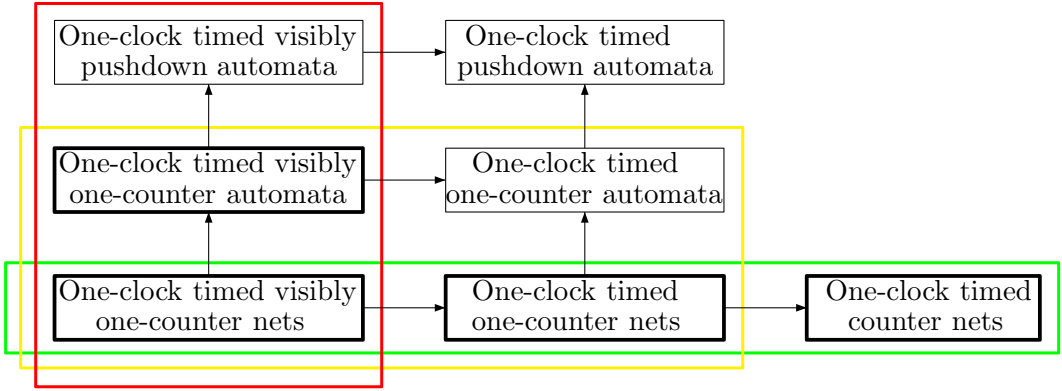

\section{Main Results}
In this section,
we present the main results of the paper.
We are interested in the language inclusion problem $L(\A)\subseteq L(\B)$, where
$\A$ and $\B$ are extensions of timed automata with discrete data structures.
Recall that according to standard notation in the field of verification, in this problem formulation 
$\B$ is seen as the \emph{specification}, and $\A$ is the system that should satisfy this specification, \ie, $\A$ should be a \emph{model} of $\B$.  
	As a special case of this problem, we consider the universality problem, \ie, the question whether $L(\B)=T\Sigma^+$ for a given automaton $\B$.	
	In general, the two problems are undecidable  for timed pushdown automata. 
	This follows on the one hand  from the undecidability of the universality problem for timed automata~\cite{AD94}, 
	and on the other hand from the undecidability of the universality problem for pushdown automata. In fact, it is long known that the universality problem is  undecidable already for non-deterministic one-counter automata~\cite{DBLP:journals/jacm/Greibach69,DBLP:journals/mst/Ibarra79}.

	However, there are interesting decidability results for subclasses of timed pushdown automata:
	The language inclusion problem is decidable if $\A$ is a timed automaton, and $\B$ is a timed automaton with at most one clock~\cite{DBLP:conf/lics/OuaknineW04}. As a special case, the universality problem for timed automata is decidable if only one clock is used. 
	The language inclusion problem is also decidable if
	$\A$ is a one-counter net and $\B$ is a finite automaton, 
	and if $\A$ is a finite automaton and $\B$ is a one-counter net~\cite{DBLP:journals/jcss/JanarEM99}.
	The universality problem for non-deterministic one-counter nets has recently been proved to have non-primitive recursive complexity~\cite{DBLP:conf/rp/HofmanT14}.
	%This also implies the decidability of the universality problem for one-counter nets.
	Further we know that the universality and language inclusion problems are decidable if $\A$ and $\B$ are visibly pushdown automata~\cite{DBLP:conf/stoc/AlurM04}.

	Hence it is interesting to consider the two problems for the corresponding subclasses of timed pushdown automata.
	It turns out that the decidability status changes depending on whether the \emph{model} uses a stack (or, more detailed: a counter) or not.
	As a first main result we present:
	\begin{theorem}
		\label{theorem_main_lang_inc}
		The language inclusion problem is undecidable if $\A$ is a timed visibly one-counter net and $\B$ is a timed automaton, even if $\A$ is deterministic and has no clocks, and $\B$ uses at most one clock. 
	\end{theorem}
	We remark that this result corrects Theorem 2 in~\cite{EmmiM06}, in which it is claimed that the language inclusion problem for the case that $\A$ is a timed pushdown automaton and $\B$ is a one-clock timed automaton is decidable.
	The incorrectness of the proof of Theorem 2 in~\cite{EmmiM06}, respectively that of Theorem 1, which the proof for Theorem 2 builds upon,  was already asserted in~\cite{DBLP:conf/concur/ChadhaV07}. Since then the problem of whether language inclusion is decidable or not has been open. 
	
	In contrast to Theorem \ref{theorem_main_lang_inc}, 
	we have decidability for the following classes: 
	%Note that $\B$ uses a counter that cannot be tested for zero: The undecidability of the universality problem for one-counter automata already implies the undecidability of the language inclusion problem if $\B$ is a timed one-counter automaton. 
	\begin{theorem}
		\label{theorem_main_dec}
		The language inclusion problem is decidable with non-primitive recursive complexity if $\A$ is a timed automaton and $\B$ is a one-clock timed counter net.
	\end{theorem}
	As a special case of this result (and with the lower bound implied by the corresponding result for one-clock timed automata~\cite{DBLP:journals/fuin/AbdullaDOQW08}), we obtain:
	\begin{corollary}
		\label{corollary_univ_dec}
		The universality problem for one-clock timed counter nets is decidable with non-primitive recursive complexity.
	\end{corollary}
	The next two sections are devoted to the proofs of Theorems \ref{theorem_main_lang_inc} and \ref{theorem_main_dec}. 
	We will also give some interesting consequences of these results and their proofs. 
	Amongst others, we prove the undecidability of the model checking problem for timed visibly one-counter nets and MTL over finite timed words, \cf \ the decidability of the same problem for timed automata~\cite{DBLP:conf/lics/OuaknineW05}. 
	After this, in Sect. 5, we will prove the following theorem:
	\begin{theorem}
		\label{theorem_univ_undec}
		The universality problem for one-clock timed visibly one-counter automata is undecidable.
	\end{theorem}
	This is in contrast to the decidability of the universality problem for the two underlying models of one-clock timed automata~\cite{DBLP:conf/lics/OuaknineW04} and visibly one-counter automata, which form a subclass of visibly pushdown automata~\cite{DBLP:conf/stoc/AlurM04}.  
	We also want to point out that this result 
	is stronger than a previous result on the undecidability of the universality problem for one-clock timed visibly pushdown automata (Theorem 3 in~\cite{EmmiM06}), and our proof closes a gap in the proof of Theorem 3 in~\cite{EmmiM06}. 
	Further, we can infer from Corollary \ref{corollary_univ_dec} and Theorem \ref{theorem_univ_undec} the exact decidability border for the universality problem of timed pushdown automata, which lies between timed visibly one-counter nets and timed visibly one-counter automata.

	\section{Undecidability Results}
	In this section, we prove Theorem \ref{theorem_main_lang_inc}. 
	The proof is a reduction of an undecidable  problem for \emph{channel machines}. 
	
	\subsection{Channel Machines}
Let $A$ be a finite alphabet. 
%We use $\varepsilon$ to denote the \emph{empty word} over $A$.
%I think i need concat for timed words! 
%Given two finite words $x,y\in A^*$, we use $x\cdot y$ to denote the \emph{concatenation} of $x$ any $y$. 
We define the order $\leq$ over the set of finite words over $A$ by 
$a_1 a_2 \dots a_m \leq b_1 b_2 \dots b_n$ if there exists a strictly increasing function $f:\{1,\dots,m\}\to\{1,\dots,n\}$ such that $a_i = b_{f(i)}$ for every $i\in\{1,\dots,m\}$. 

A \emph{channel machine} consists of a finite-state automaton acting on an  unbounded fifo channel. Formally, 
a channel machine is a tuple $\cm=(S,s_I,M,\Delta)$, where
\begin{itemize}
\item $S$ is a finite set of \emph{control states},
	\item $s_I\in S$ is the initial control state,
	\item $M$ is a finite set of \emph{messages},
	\item $\Delta \subseteq S\times L\times S$ is the transition relation over the label set $L=\{!m,?m\mid  m\in M\}\cup\{empty?\}$.
\end{itemize}
Here, $!m$ corresponds to a $send$ operation, $?m$ corresponds to a $read$ operation, and $empty?$ is a test which returns $\true$ if and only if the channel is empty.
Without loss of generality, we assume that $s_I$ does not have any incoming transitions, \ie, $(s,l,s')\in\Delta$ implies $s'\neq s_I$.
Further, we assume that $(s_I,l,s')\in\Delta$ implies $l=empty?$. 
A \emph{configuration} of $\cm$ is a pair $(s,x)$, where $s\in S$ is the control state and $x\in M^*$ represents the contents of the channel.
We use $\hs^\cm$ to denote the set of all configurations of $\cm$. 
The rules in $\Delta$ induce a transition relation $\to_{\cm}$ on $\hs^\cm\times L\times\hs^\cm$ as follows:
\begin{itemize}
\item $\langle (s,x),!m,(s',x')\rangle\hspace{-1.2mm}\in\to_\cm$ if, and only if, there exists some transition $(s,!m,s')\in\Delta$ and $x'=x\cdot m$, \ie, $m$ is added to the tail of the channel.
\item $\langle (s, x),?m,(s',x')\rangle\hspace{-1.2mm}\in\to_\cm$ if, and only if, there exists some transition $(s,?m,s')\in\Delta$ and $x=m\cdot x'$, \ie, $m$ is removed from the head of the channel. 
\item $\langle (s,x),empty?,(s',x')\rangle\hspace{-1.2mm}\in\to_\cm$ if, and only if, there exists some transition $(s,empty?,s')\in\Delta$ and $x=\varepsilon$, \ie, the channel is empty, and $x'=x$. 
\end{itemize}
We may write $(s,x)\buildrel l\over\to_\cm (s',x')$ whenever $\langle(s,x),l,(s',x')\rangle\in\to_\cm$.
Next, we define a second transition relation $\leadsto_\cm$  on $\hs^\cm\times L\times\hs^\cm$. The relation $\leadsto_\cm$ is a superset of $\to_\cm$. It contains some additional transitions which result from \emph{insertion errors}. We define $\langle (s,x_1),l,(s',x'_1)\rangle\hspace{-1.6mm}\in\leadsto_\cm$, if, and only if, there exist $x,x'\in M^*$ such that $x_1\leq x$, $\langle (s,x), l,(s',x')\rangle\hspace{-1.3mm}\in\to_\cm$, and $x'\leq x'_1$.
We may also write $(s,x)\buildrel l \over\leadsto_\cm(s',x')$ whenever $\langle(s,x),l,(s',x')\rangle\in\leadsto_\cm$. 
A \emph{computation} of $\cm$ is a finite sequence $\prod_{1\leq i\leq k} \langle (s_{i-1},x_{i-1}),l_i,(s_i,x_i)\rangle$ such that $\langle (s_{i-1},x_{i-1}),l_i,(s_i,x_i)\rangle\hspace{-1.6mm}\in\leadsto_\cm$ for every $i\in\{1,\dots,k\}$. 
We say that a computation is \emph{error-free} if for all $i\in\{1,\dots,k\}$ we have $\langle (s_{i-1},x_{i-1}),l_i,(s_i,x_i)\rangle\hspace{-1.6mm}\in\to_\cm$.
Otherwise, we say that the computation is \emph{faulty}.

The proof of Theorem \ref{theorem_main_lang_inc}  is a reduction from the following undecidable~\cite{Brand:1983:CFM:322374.322380} control state reachability problem: given a channel machine $\cm$ with control states $S$ and $s_F\in S$, does there exist an error-free computation of $\cm$ from $(s_I,\varepsilon)$ to $(s_F,x)$ for some $x\in M^*$?

We remark that the analogous problem for faulty computations is decidable with non-primitive recursive complexity~\cite{DBLP:journals/fuin/AbdullaDOQW08}: both the lower and upper bound can be proved using corresponding results for \emph{lossy channel machines}~\cite{DBLP:journals/tcs/AbdullaJ03,DBLP:journals/ipl/Schnoebelen02}.

\begin{example}
	\label{example_encoding}
	Define a channel machine $\cm=(S,M,s_I,\Delta)$ by $S=\{s_I,s,s',s_F\}$, $M=\{m_1,m_2,m_3\}$, and $\Delta=\{(s_I,empty?,s),(s,!m_1,s),(s,!m_2,s'),(s',?m_1,s'),(s',?m_3,s_F)\}$. 
	The computation	
	$$\gamma = 
	(s_I,\varepsilon)\buildrel empty?\over {\longrightarrow_\cm}  (s,\varepsilon) \buildrel !m_1 \over \longrightarrow_\cm
	(s,m_1)\buildrel !m_2\over \longrightarrow_\cm
	(s',m_1m_2)\buildrel ?m_1\over\leadsto_\cm
	(s',m_3m_2)\buildrel ?m_3\over\longrightarrow_\cm (s_F,m_2)$$ 
	is faulty due to the last but one transition, where the symbol $m_3$ is an insertion error. 	
	It is easy to see that there exists no error-free computation from $(s_I,\varepsilon)$ to $(s_F,x)$ for some $x$. 
\end{example}

The idea of our reduction is as follows: 
Given a channel machine $\cm$, we define a timed language $L(\cm)$ consisting of all timed words that encode potentially faulty computations of $\cm$ that start in $(s_I,\varepsilon)$ and end in $(s_F,x)$ for some $x\in M^*$. 
Then we define a timed visibly one-counter net $\A$ such that $L(\A)\cap L(\cm)$ contains exactly \emph{error-free} encodings of such computations. In other words, we use $\A$ to exclude the encodings of faulty computations from $L(\cm)$, 
obtaining undecidability of the non-emptiness problem for $L(\A)\cap L(\cm)$. 
Finally, we define a one-clock timed automaton $\B$ that accepts the complement of $L(\cm)$; hence the problem of deciding whether $L(\A)\not\subseteq L(\B)$ is undecidable.

\subsection{Encoding Faulty Computations}
For the remainder of Section 4, consider a channel machine 
$\cm=(S,s_I,M,\Delta)$ and let $s_F\in S$. 
Define 
$\Sigma_{\mathsf{int}} \defeq (S \backslash\{s_I\}) \cup M \cup L\cup \{\#\}$,
$\Sigma_{\mathsf{call}} \defeq \{s_I,+\}$,
and $\Sigma_{\mathsf{ret}}\defeq\{-,\star\}$,
where $+,-,\#$ and $\star$ are fresh symbols that do not occur in $S\cup M\cup L$. The symbols $+,-,\#$ are called \emph{wildcard symbols}. 
We define a timed language $L(\cm)$ over $\Sigma=\Sigma_{\mathsf{int}}\cup\Sigma_{\mathsf{call}}\cup\Sigma_{\mathsf{ret}}$ that consists of all timed words that encode %(potentially faulty) 
computations of $\cm$ from $(s_I,\varepsilon)$ to $(s_F,x)$ for some $x\in M^*$. 
The definition of $L(\cm)$ follows the ideas presented in~\cite{DBLP:conf/lics/OuaknineW05}, in which (a dual variant of) the 
control state reachability problem for channel machines is reduced to the satisfiability problem for MTL over timed words.

In general, the idea is to encode a configuration of $\cm$ of the form $(s, x)$ by a timed word of duration one. 
This timed word starts with the symbol $s$ at some time $t$.
If the content of the channel $x$ is of the form $m_1 m_2 \dots m_j$, 
then $s$ is followed by the symbols $m_1, m_2,\dots,m_j$ in this order. The timestamps of these symbols must be strictly monotonic and in the interval $(t,t+1)$. 
Due to the denseness of the time domain, one can indeed store the channel content in one time unit without any upper bound on $j$.

\begin{example}
	The initial configuration $(s,\varepsilon)$ is encoded by a
        single-letter timed word, for instance by $(s,1.0)$. 
The configuration $(s',m_1m_2)$ may be encoded by the timed word $(s',7.0)(m_1,7.2)(m_2,7.8)$. 
The choice of the timestamps is arbitrary as long as the timestamps of the message symbols are in the unit interval determined by the timestamp of the preceding control state. 
\end{example}

%Now, let $\gamma = \prod_{1\leq i\leq k} \langle (s_{i-1},x_{i-1}),l_i,(s_i,x_i)\rangle$ be a computation of $\cm$ with $s_0=s_I$, $x_0=\varepsilon$, and $s_k=s_F$. 
To encode a computation of a channel machine, we concatenate the encodings of each of the participating configurations in the following way. 
Encodings of consecutive configuration have a time distance of exactly two time units. 
One time unit before the encoding of the next configuration we store the label of the transition. 
% weglassen?
%The encoding further relates message symbols occurring in the channel of some configuration to the message symbols occurring in the channel of the next configuration. % according to the first-in-first-out-policy of the channel machine. 
Further, each message symbol in the encoding of the current configuration has a matching copy in the encoding of the next configuration,  after exactly two time units, and in accordance with the following rules: %, unless the symbol at the head of the channel is removed from the channel by a read operation. 
If the next transition is sending a new message symbol $m$ to the tail of the channel, we add $m$ to the tail of the encoding of the next configuration. 
If the next transition is reading $m$ from the head of the channel, we test whether the first symbol after the control state in the encoding of the current configuration equals $m$, and if so, we remove it from the encoding of the next configuration.  
\begin{example}
	The transition $(s',m_1m_2)\buildrel ?m_1\over\longrightarrow_\cm(s',m_2)$ 
 may be encoded by $$(s',7.0)(m_1,7.2)(m_2,7.8)(?m_1,8.0)(s',9.0)(m_2,9.8).$$ 
 The symbol $m_2$ at time $7.8$  has a matching copy exactly two time units later; the symbol $m_1$ at $7.2$ does not, because it is removed from the channel by the transition. 
\end{example}
This idea of encoding computations of channel machines was used for proving lower complexity bounds for the satisfiability problem of MTL~\cite{DBLP:conf/lics/OuaknineW05} and the universality problem for one-clock timed automata~\cite{DBLP:journals/fuin/AbdullaDOQW08}. 
These problems, however, are decidable~\cite{DBLP:conf/lics/OuaknineW05,DBLP:journals/fuin/AbdullaDOQW08}, whereas here we want to use the encoding to show undecidability of a problem. The crucial point is that the encoding explained above does not exclude timed words that are encoding \emph{faulty} computations of a channel machine: for excluding encodings of faulty computations, we need to require that every message symbol has not only a matching copy \emph{after} two time units, but also a matching copy two time units \emph{before}. In other words, it should not be possible that message symbols appear ``all of a sudden'', \ie, without a corresponding error-free transition.
\begin{example}
	The faulty transition $(s',m_1m_2)\buildrel ?m_1\over\leadsto_\cm(s',m_3m_2)$ 
	may be encoded by $$(s',7.0)(m_1,7.2)(m_2,7.8)(?m_1,8.0)(s',9.0)(m_3,9.1)(m_2,9.8).$$ 
	As in the previous example, the symbol $m_2$ at time $7.8$  has a matching copy exactly two time units later; the symbol $m_1$ at $7.2$ does not, because it is removed from the channel by the encoded transition. 
	However, the symbol $m_3$ at time $9.1$ appears all of a sudden, \ie, without any matching copy two time units before.
 This corresponds to an insertion error in the computation, see Example \ref{example_encoding}. 
\end{example}
The above described \emph{backward-looking conditions}, however, cannot be expressed by neither MTL formulas, nor by one-clock timed automata.\footnote{Backward-looking conditions (or, to be more exact with respect to the reduction: the violation of such backward-looking conditions), can be expressed by MTL with past operators, and by timed automata with two clocks, leading to the undecidability of the corresponding satisfiability and universality problem, respectively~\cite{AD94}.} Due to this failure, it is only the control state reachability problem \emph{for faulty computations} that can be reduced to the satisfiability problem for MTL respectively the universality problem for one-clock timed automata. As mentioned before, the control state reachability problem for faulty computations is \emph{decidable}~\cite{DBLP:journals/fuin/AbdullaDOQW08}.

For our undecidability proof to work,
we have to exclude encodings of faulty computations. 
In other words, we have to exclude timed words in which message symbols occur without any matching copy two time units before. 
This will be carried out by the counter of the visibly timed one-counter net $\A$. 
For this to work, we have to change the encoding in some details, as explained in the following.

%Let us for a moment assume that the channel machine knows in advance the maximum length $n\in \N$ of the channel content during an error-free computation. 
%Instead of starting with an empty channel, 
%the channel contains the word $\#^n$, \ie, $n$ occurrences of the wildcard symbol $\#$. 

%Let $\gamma= \prod_{1\leq i\leq k} \langle (s_{i-1},x_{i-1}),l_i,(s_i,x_i)\rangle$ be an error-free computation of $\cm$. 
%Assume we want to encode $\gamma$. 
%We define $\max(\gamma)$ to denote the maximum length of the channel content occurring in $\gamma$, formally: $\max(\gamma)\defeq\max\{|x_i| \mid 0\leq x_i\leq k\}$. 

Assume we want to encode a given error-free computation of $\cm$. 
Let $n$ be the maximum length of the channel content during this computation. 
Let us assume for a moment that $\cm$ does not start its computation with the empty channel, but the channel contains the word $\#^n$, \ie, $n$ occurrences of the wildcard symbol $\#$. The semantics of $\cm$ is changed in the following way: If a message $m$ is sent, then, instead of adding $m$ to the tail of the channel, the first $\#$ occurring in the channel is \emph{replaced} by $m$. 
If a message $m$ is read, then, it is tested whether the first symbol in the channel is $m$. If this the case, it is removed from the channel and additionally a new wildcard symbol $\#$  is added to the tail of the channel. 
A test for emptiness of the channel is replaced by testing whether the channel only contains the wildcard symbol $\#$. 
\begin{example}
	Let $(s,m_1m_1)\buildrel!m_2\over\longrightarrow_\cm(s',m_1m_1m_2)\buildrel ?m_1\over\longrightarrow_\cm(s',m_1m_2)$ be a computation of $\cm$, and let $n=4$. 
	With the new wildcard semantics this computation is of the form
	$$(s,m_1m_1\#\#)\buildrel!m_2\over\longrightarrow_\cm(s',m_1m_1m_2\#)\buildrel ?m_1\over\longrightarrow_\cm(s',m_1m_2\#\#).$$
\end{example}
Observe that the length of the channel content is constantly $n$. 
We will later exploit this fact.

The encoding of computations of channel machines is now changed with this wildcard semantics in mind. 
The initial configuration is encoded by a timed word that starts with the symbol $s_I$ at some time $t$, and then is followed by $n$ occurrences of $\#$, all of which have monotonically increasing timestamps in the interval $(t,t+1)$. 
The rules for the transitions change accordingly:
If the next transition is sending a message $m$ to the channel, we \emph{replace} in the encoding of the next configuration the first occurrence of $\#$ by $m$. 
If the next transition is reading a message $m$ from the head of the channel, we test whether the first symbol after the control state is $m$, and if so, we remove it from the encoding of the next configuration. We further add a new $\#$ to the end of the encoding of the next configuration.%, and we shift the times of all symbols. 
\begin{example}
	The encoding of the computation of the previous example may be of the form
	$$(s,7.0)(m_1,7.2)(m_1,7.5)(\#,7.6)(\#,7.8)(!m_2,8.0)(s',9.0)(m_1,9.2)(m_1,9.5)(m_2,9.6)(\#,9.8)$$
	$$(?m_1,10.0)
	(s',11.0)(m_1,11.5)(m_2,11.6)(\#,11.8)(\#,11.9).$$
\end{example}
Observe that the length of the encoding of every configuration is constantly $n+1$. This is due to the fact that none of the encoding rules changes the number of symbols in a configuration. 
However, due to the lack of backward-looking conditions, 
it may still happen that some symbols appear ``all of a sudden'', \ie, without a matching copy two time units before. 
In this case, the length of the encoding of the configuration increases. By the encoding conditions above, the increasing effect will be carried over to the encodings of the following configurations. 
Hence, in order to find out whether insertion errors occurred, 
it suffices to compare the length of the encoding of the initial configuration (which is equal to $n+1$) to the length of the encoding of the last configuration: If it is still $n+1$, then we know that no insertion error has occurred; otherwise, some insertion error has occurred. 
The test will be done by the counter of the timed visibly one-counter net $\A$. 
During a run of $\A$ on the encoding of a computation, 
the counter is incremented while the symbols of the encoding of the first configuration are read, and it is decremented while the symbols of the encoding of the last configuration are read. In between, the counter is not touched. By the fact that a one-counter net cannot decrement the counter more often than it was incremented ($\A$ is blocked as soon as the counter would become negative), we can exclude timed words which are encoding potential insertion errors. 

By definition of timed \emph{visibly} one-counter nets, $\A$ is restricted to use $s_I$ and the wildcard symbol $+$ whenever the counter should be incremented, and it can only use the wildcard symbol $-$ and $\star$ whenever the counter should be decremented.  
This requires some extra effort in the encoding, namely that the encoding of the first configuration only uses the wildcard symbol $+$, and the encoding of the last configuration only uses the wildcard symbol $-$. 
Before we give the formal definition of the language $L(\cm)$, 
we show a complete encoding of the faulty computation of Example \ref{example_encoding} for $n=2$.
\begin{example}
	\label{example_thewholeenc}
	$$(s_I,1.0)(+,1.2)(+,1.8)(empty?,2.0)(s,3.0)(\#,3.2)(\#,3.8)(!m_1,4.0)(s,5.0)(m_1,5.2)(\#,5.8)$$
	$$(!m_2,6.0)(s',7.0)(m_1,7.2)(m_2,7.8)(?m_1,8.0)(s',9.0)(m_3,9.1)(m_2,9.8)(\#,9.9)(?m_3,10.0)$$
	$$(s_F,11.0)(-,11.8)(-,11.9)(-,11.95)(\star,12.0)$$
\end{example}
For every $n\in\N$, we define a timed language $L(\cm, n)$  as follows: 
The timed language $L(\cm,n)$ consists of all timed words $w$ over $\Sigma$ that satisfy the following conditions:
\begin{enumerate}
\item $w$ must be strictly monotonic. 
	
\item The untiming of $w$ must be of the form given by the following regular expression: 
	$$s_I (+)^n empty? (S M^* \#^* L)^* s_F -^* \star$$
\item For every $s\in S$, if $s$ is followed by $l$ after one time unit, and $s$ is followed by $s'$ after two time units, then $(s,l,s')\in\Delta$.

\item For every $s\in S$ with $s\neq s_F$, there exists $l\in L$ after exactly one time unit. 
	
\item For every $s\in S$ with $s\neq s_F$, there exists $s'\in S$ after exactly two time units.
	
\item After  $s_F$ the symbol $\star$ occurs after two time units.

\end{enumerate}
Further,
for every infix of $w$ of the form
$$(s,\delta)(\sigma_1,\delta_1)(\sigma_2,\delta_2)\dots(\sigma_k,\delta_k)(l,\delta+1)(s',\delta+2)(\sigma'_1,\delta'_1)\dots(\sigma'_{k'},\delta'_{k'})(l',\delta+3)$$
with $s,s'\in S\backslash\{s_F\}$,
$l,l'\in L$,
$\delta\in\RP$,
$\delta<\delta_1<\dots<\delta_k<\delta+1<\delta+2<\delta'_1<\dots<\delta'_{k'}<\delta+3$,
there exists a strictly increasing function $f:\{1,\dots,k\}\to\{1,\dots,k'\}$ such that the following conditions are satisfied:
\begin{enumerate}
	\addtocounter{enumi}{6} 
\item If $l=empty?$,
	then 
	\begin{enumerate}
	\item $\sigma_i\in\{+,\#\}$ and $\sigma'_{f(i)}=\#$ for all $i\in\{1,\dots,k\}$ (the channel is empty), and
	\item $\delta'_{f(i)}=\delta_i+2$ for every $i\in\{1,\dots,k\}$ (there are matching copies after two time units).
	\end{enumerate}	
%\end{enumerate}
\item If $l=!m$ for some $m\in M$, then
	\begin{enumerate}
	\item there exists $i\in\{1,\dots,k\}$ such that $\sigma_i=\#$ (there is some wildcard symbol in the encoding of the current configuration),  
	\item $\sigma'_{f(j)}=m$ and $\delta'_{f(j)}=\delta_j+2$ for $j=\min\{i\in\{1,\dots,k\}\mid \sigma_i=\#\}$ (the first wildcard symbol is replaced by $m$ two time units later), 
	\item $\sigma'_{f(i)}=\sigma_i$ and $\delta'_{f(i)} = \delta_i+2$  for all $i\in\{1,\dots,k\}\backslash\{j\}$ (there are matching copies for the remaining symbols).
	\end{enumerate}
\item If $l=?m$ for some $m\in M$, then 
	\begin{enumerate}
	\item $\sigma_1=m$ (the first symbol is equal to $m$),
	\item $f(k)=k'$, $\sigma'_{f(k)}=\sigma_{k'}=\#$, and 
		$\delta'_{f(k)}=\delta'_{k'}>2+\delta_k$ (a new wildcard symbol is added at the end of the encoding), and
	\item $\sigma'_{f(i)}=\sigma_{i+1}$ and $\delta'_{f(i)}=2+\delta_{i+1}$  for all $i\in\{1,\dots,k-1\}$ (there are matching copies for all other symbols except for the first one, which is removed from the encoding). 
\end{enumerate}		
\end{enumerate}\smallskip

\noindent For an infix like above but with $s'=s_F$ and $l'=\star$, we add conditions (7'), (8') and (9') that   differ from 7, 8, and 9, respectively, in that all message or wildcard symbols to be copied or added to the encoding of the next configuration are replaced by the wildcard symbol $-$.

A simple observation that we will later use is that by these conditions the length of the encodings of two consecutive configurations cannot decrease. 
Indeed, the conditions for representing transitions between two configurations of the channel machine do not change the number of symbols in the encoding of the respective configurations. 
By the lack of backward-looking conditions it may however happen that some symbols appear all of a sudden, \ie, without a matching copy two time units before. We point out that such insertion errors \emph{may} occur, but they are not required by any of the conditions. 
%THIS IS ONLY TRUE FOR error-free comps%We point out that for every $n$ there is a timed word in $L(\cm,n)$ such that no such insertion errors occur and thus the length of the encoding of all configurations, and in particular the last one, is equal to $n$. 
	
\subsection{Excluding Faulty Computations}
We define a timed visibly one-counter net $\A$ over $\Sigma$ such that for every $n\in\N$ the intersection $L(\A)\cap L(\cm,n)$ consists of all timed words that encode \emph{error-free} computations of $\cm$ from $(s_I,\varepsilon)$ to $(s_F,x)$ for some $x\in M^*$. 
The timed visibly one-counter net $\A$ is shown in Figure \ref{figure_A}.
After incrementing the counter while reading the initial letter $s_I$, it non-deterministically guesses a number $n\in \N$ of symbols $+$ and increments the counter each time it reads the symbol $+$.
When $\A$ leaves $\loc_1$, the value of the counter is $n+1$. After that, the counter value is not changed until the state symbol $s_F$ is read.
Then, while reading symbols in $\{-,\star\}$, the counter value is decremented. Note that $\A$ can reach the final location $\loc_4$ only if the number of the occurrences of symbol $-$ between $s_F$ and $\star$ is \emph{at most} $n$: otherwise the counter value would become negative, and thus the edges going out from $\loc_3$ would be blocked. 
Note that $\A$ does not use any clock, and it is deterministic. 
	\begin{figure}
\begin{center}
		\begin{picture}(107,14)(0,-14)
%\put(0,-22){\framebox(110,22){}}
\node[NLangle=0.0,Nmarks=i,ilength=3,Nw=4.0,Nh=4.0,Nmr=2.0](n0)(5.0,-11.0){$\loc_0$}
\node[NLangle=0.0,Nw=4.0,Nh=4.0,Nmr=2.0](n1)(27.0,-11){$\loc_1$}
\node[NLangle=0.0,Nw=4.0,Nh=4.0,Nmr=2.0](n2)(50.0,-11){$\loc_2$}
\drawedge[curvedepth=4.0](n0,n1){\footnotesize{$s_I,push$}}
\drawloop[loopdiam=6](n1){\footnotesize{$+,push$}}
\drawedge[curvedepth=4.0](n1,n2){\footnotesize{$L,noop$}}
\drawloop[loopdiam=6](n2){\footnotesize{$\Sigma\backslash\{s_F\},noop$}}
\node[NLangle=0.0,Nw=4.0,Nh=4.0,Nmr=2.0](n3)(76.0,-11){$\loc_3$}
\node[NLangle=0.0,Nmarks=f,flength=3,Nw=4.0,Nh=4.0,Nmr=2.0](n4)(100.0,-11){$\loc_4$}
\drawedge[curvedepth=4.0](n2,n3){\footnotesize{$s_F,noop$}}
\drawloop[loopdiam=6](n3){\footnotesize{$-, pop$}}
%\drawloop[loopdiam=6,loopangle=270](n3){\footnotesize{$\Sigma \backslash (M\cup \{\star,\#\}), noop $}}
\drawedge[curvedepth=4.0](n3,n4){\footnotesize{$\star, pop$}}
\end{picture}
\caption{The deterministic timed visibly one-counter net $\A$ for excluding insertion errors. }
\label{figure_B}
\end{center}
\end{figure}
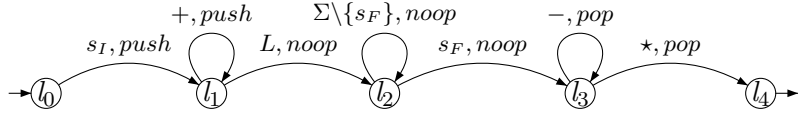
\begin{example}
	The timed word presented in Example \ref{example_thewholeenc} cannot be accepted by the timed visibly one-counter net $\A$ in Figure \ref{figure_B}: 
	The run of $\A$ on the prefix is of the form 
	$$\langle(\loc_0,0),1.0,s_I,(\loc_1,1)\rangle\langle(\loc_1,1),0.2,+,(\loc_1,2)\rangle\langle(\loc_1,2),0.6,+,(\loc_1,3)\rangle\langle(\loc_1,3),0.2,empty?,(\loc_2,3)\rangle.$$
	Here, in $(\loc,c)$, $\loc$ stands for the current location, and $c$ stands for the current value of the counter. 
	The counter value stays constant until we finally read the first $-$:  
	$$\langle(\loc_2,3),1.0,s_F,(\loc_3,3)\rangle\langle(\loc_3,3),0.8,-,(\loc_3,2)\rangle\langle(\loc_3,2),0.1,-,(\loc_3,1)\rangle\langle(\loc_3,1),0.05,-,(\loc_3,0)\rangle.$$
	Now $\A$ is blocked:
	all outgoing edges of $\loc_3$ require the counter to be decremented, which is not possible if the counter has value zero. 
	This is indeed what we want: the timed word in Example \ref{example_encoding} is encoding a faulty computation, and should be excluded.  
\end{example}

\begin{lemma}
	\label{lemma_intersection}
	$\cm$ has an error-free computation from $(s_I,\varepsilon)$ 
	to $(s_F,x)$ for some $x\in M^*$, if, and only if, 
	there exists $n\in \N$  such that  $L(\cm,n)\cap L(\A)\neq \emptyset$.
\end{lemma}
\begin{proof}
	For the direction from left to right,
	let $\gamma$ be an error-free computation of $\cm$ from $(s_I,\varepsilon)$ to $(s_F,x)$ for some $x\in M^*$. 
	Let $n$ be the maximum length of the channel content during $\gamma$. 
	Let $w$ be a timed word in $L(\cm,n)$ in which no message or wildcard symbols occur ``all of a sudden'', \ie, without a matching copy two time units before. 
	Note that such a timed word exists, because $\gamma$ is error-free, and hence there is no need to encode insertion errors into $w$. 
	This implies that the length of the encodings of all, and in particular, the first and the last configuration in $w$ is $n+1$. 
	This implies $w\in L(\A)$.
	Hence $L(\cm,n)\cap L(\A)\neq\emptyset$.

	For the direction from right to left, 
	let $w\in L(\cm,n)\cap L(\A)$ for some $n\in\N$. 
	%By definition of $\A$, 
	%the value of the counter after reading the (untimed) prefix $s_I +^n$ 
	%equals $n+1$. 
	By definition of $L(\cm,n)$, 
	the length of the encoding of the initial configuration is $n+1$. 
	By the observation above, the length of the  encoding of the last configuration  is thus at least $n+1$, too. 	
	However, by definition of $\A$, 
	the length of the encoding of the last configuration cannot be greater than $n+1$, because otherwise the edge to $\loc_4$ cannot be taken due to the decrement operation. 
	By the observation above, the length of encodings of consecutive configurations do not decrease, and thus the length of the encodings of all configurations is $n+1$. Hence we can conclude that there are no insertion errors necessary to encode the execution of a transition. This implies that 
	there exists some error-free computation $\gamma$ of $\cm$ from $(s_I,\varepsilon)$ to $(s_F,x)$ for some $x\in M^*$. 	
\end{proof}
We finally define $L(\cm)\defeq\bigcup_{n\in\N} L(\cm,n)$.
	\begin{corollary}
		\label{corollary_cap}
		There exists some error-free computation of $\cm$ from $(s_I,\varepsilon)$ to $(s_F,x)$ for some $x\in M^*$ if, and only if, $L(\A)\cap L(\cm)\neq \emptyset$. 
	\end{corollary}

	\subsection{The Reduction}
	Finally, 
	we define a one-clock timed automaton $\B$ such that $L(\B)=\TW \backslash L(\cm)$. 
	The construction of $\B$ follows the same ideas as, \eg, in~\cite{DBLP:conf/formats/AdamsOW07}: 
	$\B$ is the union of several one-clock timed automata, each of them violating some condition of the definition of $L(\cm)$, as described in the following.

	The timed automaton in Figure A accepts timed words that are not strictly monotonic,  thus violating condition (1). 
	For accepting timed words violating condition (2), we can construct a finite automaton that recognizes the complement of the given regular expression. 
	Define 
	for every $s\in S\backslash\{s_f\}$ the sets $\overline{L}(s)=\{l\in L\mid \neg\exists s'\in S. (s,l,s')\in\Delta\}$ and $\overline{S}(s)=\{s'\in S\mid \neg\exists l\in L. (s,l,s')\in\Delta\}$, and $\overline{L}(s_F)=\Sigma\backslash\{\star\}$ and $\overline{S}(s_F)=\Sigma$. 
	Then for every $s$ the corresponding timed automaton in Figure B accepts timed words that contain the encoding of a transition $(s,l,s')\not\in\Delta$, thus violating condition (3).
	Violations of the forward-looking conditions in (4), (5), and (6) can be 
	accepted by the timed automaton in Figure C with, respectively, 
	(4) $M_1=S$, $M_2=L$, and $k=1$, (5) $M_1=S$, $M_2=S$, and $k=2$, and (6) $M_1=\{s_F\}$, $M_2=\{\star\}$, and $k=1$.

	\begin{center}
		\begin{tabular}{p{5cm}p{1cm}p{5cm}}
			\\
\mbox{\begin{picture}(48,22)(0,-22)
%\put(0,-20){\framebox(48,20){}}
\node[NLangle=0.0,Nmarks=i,ilength=3,Nw=4.0,Nh=4.0,Nmr=2.0](n0)(5.0,-13.0){}
\node[NLangle=0.0,Nw=4.0,Nh=4.0,Nmr=2.0](n1)(25.0,-13){}
\node[NLangle=0.0,Nmarks=f,flength=3,Nw=4.0,Nh=4.0,Nmr=2.0](n2)(43.0,-13){}
\drawloop[loopdiam=4](n0){\footnotesize{$\Sigma$}}
\drawloop[loopdiam=4](n2){\footnotesize{$\Sigma$}}
\drawedge[curvedepth=4.0](n0,n1){\footnotesize{$\Sigma, x\defeq 0$}}
\drawedge[curvedepth=4.0](n1,n2){\footnotesize{$\Sigma, x= 0$}}
\end{picture}} & & 
\begin{picture}(50,22)(0,-22)
%\put(0,-22){\framebox(50,22){}}
\node[NLangle=0.0,Nmarks=i,ilength=3,Nw=4.0,Nh=4.0,Nmr=2.0](n0)(5.0,-13.0){}
\node[NLangle=0.0,Nw=4.0,Nh=4.0,Nmr=2.0](n1)(21.0,-13){}
\node[NLangle=0.0,Nmarks=f,flength=3,Nw=4.0,Nh=4.0,Nmr=2.0](n2)(45.0,-13){}

\drawedge[curvedepth=4.0](n0,n1){\footnotesize{$s,x\defeq 0$}}
\drawloop[loopdiam=4](n0){\footnotesize{$\Sigma$}}

\drawloop[loopdiam=4](n2){\footnotesize{$\Sigma$}}
\drawedge[curvedepth=4.0](n1,n2){\footnotesize{$\overline{L}(s), x=1$}}
\drawedge[curvedepth=-4.0,ELside=r](n1,n2){\footnotesize{$\overline{S}(s), x=2$}}

\end{picture}
 \\
 Figure A: Condition (1) & &Figure B: Condition (3)
\end{tabular}\end{center}

\begin{center}
		\begin{tabular}{p{8cm}}
		\\
		\mbox{
 \begin{picture}(75,20)(0,-20)
%\put(0,-20){\framebox(75,20){}}
\node[NLangle=0.0,Nmarks=i,ilength=3,Nw=4.0,Nh=4.0,Nmr=2.0](n0)(5.0,-13.0){}
\node[NLangle=0.0,Nw=4.0,Nh=4.0,Nmr=2.0](n1)(30.0,-13){}
\node[NLangle=0.0,Nmarks=f,flength=3,Nw=4.0,Nh=4.0,Nmr=2.0](n2)(70.0,-13){}
\drawloop[loopdiam=6](n0){\footnotesize{$\Sigma$}}
\drawloop[loopdiam=6](n2){\footnotesize{$\Sigma$}}
\drawedge[curvedepth=4.0](n0,n1){\footnotesize{$M_1, x:=0$}}
\drawloop[loopdiam=6](n1){\footnotesize{$\Sigma\backslash M_2, x<k$}}

\drawedge[curvedepth=6.0](n1,n2){\footnotesize{$M_2, x<k$}}
\drawedge[curvedepth=0.0](n1,n2){\footnotesize{$\Sigma\backslash M_2, x=k$}}
\drawedge[curvedepth=-6.0](n1,n2){\footnotesize{$\Sigma, x>k$}}
\end{picture}}
 \\
 Figure C: Conditions (4), (5), and (6)
 \\
 \\
\end{tabular}\end{center}
Timed words violating condition (7a) can be accepted by the timed automaton in Figure D. 
In Figure E we show a timed automaton that accepts timed words violating condition (7b).

	 \begin{tabular}{p{4.1cm}p{10cm}}
\mbox{
	\begin{picture}(41,20)(0,-20)
%\put(0,-20){\framebox(41,20){}}
\node[NLangle=0.0,Nmarks=i,ilength=3,Nw=4.0,Nh=4.0,Nmr=2.0](n0)(0.0,-13.0){}
\node[NLangle=0.0,Nw=4.0,Nh=4.0,Nmr=2.0](n1)(14.0,-13){}
\node[NLangle=0.0,Nmarks=f,flength=3,Nw=4.0,Nh=4.0,Nmr=2.0](n2)(31.0,-13){}
%\node[NLangle=0.0,Nmarks=f,flength=3,Nw=4.0,Nh=4.0,Nmr=2.0](n3)(62.0,-13){}
\drawloop[loopdiam=4](n0){\footnotesize{$\Sigma$}}
\drawloop[loopdiam=4](n2){\footnotesize{$\Sigma$}}
%\drawloop[loopdiam=4](n3){\footnotesize{$\Sigma$}}
\drawedge[curvedepth=4.0](n0,n1){\footnotesize{$M$}}
\drawloop[loopdiam=4](n1){\footnotesize{$M,\#$}}
\drawedge[curvedepth=4.0](n1,n2){\footnotesize{$\empty?$}}
\end{picture}}
&\begin{picture}(95,20)(0,-20)
%\put(0,-20){\framebox(95,20){}}
\node[NLangle=0.0,Nmarks=i,ilength=3,Nw=4.0,Nh=4.0,Nmr=2.0](n0)(5.0,-13.0){}
\node[NLangle=0.0,Nw=4.0,Nh=4.0,Nmr=2.0](n1)(30.0,-13){}
\node[NLangle=0.0,Nw=4.0,Nh=4.0,Nmr=2.0](n3)(55.0,-13){}
\node[NLangle=0.0,Nmarks=f,flength=3,Nw=4.0,Nh=4.0,Nmr=2.0](n2)(90.0,-13){}
\drawloop[loopdiam=4](n0){\footnotesize{$\Sigma$}}
\drawloop[loopdiam=4](n2){\footnotesize{$\Sigma$}}
\drawloop[loopdiam=4](n1){\footnotesize{$+,\#$}}
\drawedge[curvedepth=4.0](n0,n1){\footnotesize{$+,\#, x:=0$}}
\drawedge[curvedepth=4.0](n1,n3){\footnotesize{$empty?$}}

\drawloop[loopdiam=4](n3){\footnotesize{$\Sigma, x<2$}}

\drawedge[curvedepth=3.0](n3,n2){\footnotesize{$\Sigma\backslash\{\#\}, x=2$}}
\drawedge[curvedepth=-3.0](n3,n2){\footnotesize{$\Sigma, x>2$}}
\end{picture}
\\
\hspace{-5mm}Figure D: Condition (7a) & Figure E: Condition (7b) 
\\
\\
\end{tabular}
%The next timed automata are used to accept timed words violating conditions (8). 
%For every $m\in M$, we define:
The timed automaton in Figure F accepts all timed words for which the last symbol before $!m$ is different from $\#$. This, together with the structure of $w$ ensured by condition (2), implies that condition (8a) is violated. 
The timed automaton in Figure G accepts timed words violating condition (8b). 
For condition (8c), we can use a timed automaton similar to that in Figure E.
%we can reuse the timed automaton in Figure E with $M_1=M_2=\{\#\}$, $l=!m$, and $\tau=\#$, preceded by a pair of transitions labelled by $M\cup S$ and then $\#$. 
By constructing for every $m\in M$ the corresponding automata, we can thus accept all timed words violating the conditions stated in (8). 
	\begin{center}
		\begin{tabular}{p{2.9cm}p{5mm}p{8.5cm}}
\mbox{\begin{picture}(30,20)(0,-19)
%\put(0,-19){\framebox(30,20){}}
\node[NLangle=0.0,Nmarks=i,ilength=3,Nw=4.0,Nh=4.0,Nmr=2.0](n0)(5.0,-13.0){}
\node[NLangle=0.0,Nw=4.0,Nh=4.0,Nmr=2.0](n1)(16.0,-13){}
\node[NLangle=0.0,Nmarks=f,flength=3,Nw=4.0,Nh=4.0,Nmr=2.0](n2)(25.0,-13){}
\drawloop[loopdiam=4](n0){\footnotesize{$\Sigma$}}
%\drawloop[loopdiam=4](n1){\footnotesize{$M$}}
\drawloop[loopdiam=4](n2){\footnotesize{$\Sigma$}}
\drawedge[curvedepth=4.0](n0,n1){\footnotesize{$\ \ \Sigma\backslash\{\#\}$}}
\drawedge[curvedepth=4.0](n1,n2){\footnotesize{$!m$}}
\end{picture}}& &\begin{picture}(83,19)(0,-19)
%\put(0,-19){\framebox(83,19){}}

\node[NLangle=0.0,Nmarks=i,ilength=3,Nw=4.0,Nh=4.0,Nmr=2.0](n0)(5.0,-13.0){}
\node[NLangle=0.0,Nw=4.0,Nh=4.0,Nmr=2.0](n1)(18.0,-13){}
\node[NLangle=0.0,Nw=4.0,Nh=4.0,Nmr=2.0](n2)(34.0,-13){}
\node[NLangle=0.0,Nw=4.0,Nh=4.0,Nmr=2.0](n3)(45.0,-13){}
\node[NLangle=0.0,Nw=4.0,Nh=4.0,Nmr=2.0](n4)(54.0,-13){}
\node[NLangle=0.0,Nmarks=f,flength=3,Nw=4.0,Nh=4.0,Nmr=2.0](n5)(78.0,-13){}

\drawloop[loopdiam=4](n0){\footnotesize{$\Sigma$}}
\drawedge[curvedepth=3.0](n0,n1){\footnotesize{$S\cup M$}}
%\drawloop[loopdiam=4](n1){\footnotesize{$M$}}
\drawedge[curvedepth=3.0](n1,n2){\footnotesize{$\#,x\defeq 0\ $}}
\drawloop[loopdiam=4](n2){\footnotesize{$\#$}}
\drawedge[curvedepth=3.0](n2,n3){\footnotesize{$!m$}}
\drawedge[curvedepth=3.0](n3,n4){\footnotesize{$S$}}

\drawedge[curvedepth=-4.0](n4,n5){\footnotesize{$\Sigma,x\neq 2$}}
\drawedge[curvedepth=4.0](n4,n5){\footnotesize{$\Sigma\backslash\{m\},x=2$}}
\drawloop[loopdiam=4](n5){\footnotesize{$\Sigma$}}
\end{picture} \\
Figure F: Condition (8a)& &Figure G: Condition (8b) \\
\\ 
\end{tabular}	
\end{center}
Last but not least we present timed automata that accept timed words violating condition (9). 
For all $m\in M$, we define  
timed automata shown in Figure H  und I, respectively, that accept timed words violating condition (9a) and (9b), respectively.
For (9c) we construct a timed automaton very similar to that in Figure E. 
	 \begin{center}
\begin{tabular}{p{5cm}p{8cm}}
	\\
\mbox{\begin{picture}(46,19)(0,-19)
%\put(0,-19){\framebox(46,19){}}

\node[NLangle=0.0,Nmarks=i,ilength=3,Nw=4.0,Nh=4.0,Nmr=2.0](n0)(5.0,-13.0){}
\node[NLangle=0.0,Nw=4.0,Nh=4.0,Nmr=2.0](n1)(16.0,-13){}
\node[NLangle=0.0,Nw=4.0,Nh=4.0,Nmr=2.0](n2)(30.0,-13){}
\node[Nmarks=f,flength=3,NLangle=0.0,Nw=4.0,Nh=4.0,Nmr=2.0](n3)(41.0,-13){}

\drawloop[loopdiam=4](n0){\footnotesize{$\Sigma$}}

\drawedge[curvedepth=3.0](n0,n1){\footnotesize{$S$}}
\drawedge[curvedepth=3.0](n1,n2){\footnotesize{$\Sigma\backslash\{m\}$}}
\drawloop[loopdiam=4](n2){\footnotesize{$M,\#$}}
\drawloop[loopdiam=4](n3){\footnotesize{$\Sigma$}}
\drawedge[curvedepth=3.0](n2,n3){\footnotesize{$?m$}}
\end{picture}}&\begin{picture}(80,19)(0,-19)
%\put(0,-19){\framebox(80,19){}}
\node[NLangle=0.0,Nmarks=i,ilength=3,Nw=4.0,Nh=4.0,Nmr=2.0](n0)(5.0,-13.0){}
\node[NLangle=0.0,Nw=4.0,Nh=4.0,Nmr=2.0](n1)(28.0,-13){}
\node[NLangle=0.0,Nw=4.0,Nh=4.0,Nmr=2.0](n2)(40.0,-13){}
\node[NLangle=0.0,Nw=4.0,Nh=4.0,Nmr=2.0](n3)(60.0,-13){}
\node[NLangle=0.0,Nmarks=f,flength=3,Nw=4.0,Nh=4.0,Nmr=2.0](n4)(75.0,-13){}

\drawloop[loopdiam=4](n0){\footnotesize{$\Sigma$}}

\drawedge[curvedepth=4.0](n0,n1){\footnotesize{$M,\#,x\defeq 0$}}

\drawedge[curvedepth=4.0](n1,n2){\footnotesize{$?m$}}
\drawloop[loopdiam=6](n2){\footnotesize{$\Sigma, x\le 2$}}

\drawedge[curvedepth=4.0](n2,n3){\footnotesize{$M, x>2$}}
\drawedge[curvedepth=-4.0](n2,n4){\footnotesize{$L$}}

%\drawedge[curvedepth=-4.0](n2,n4){\footnotesize{$M, x>2$}}
\drawloop[loopdiam=4](n3){\footnotesize{$M$}}
\drawedge[curvedepth=4.0](n3,n4){\footnotesize{$L$}}
\drawloop[loopdiam=4](n4){\footnotesize{$\Sigma$}}
\end{picture} \\
Figure H: Condition (9a)& Figure I: Condition (9b) \\
\\
\end{tabular}\end{center}
So let $\B$ the union of all these timed automata. 
One can easily see that for every timed word $w\in\TW$ we have $w\in L(\B)$ if, and only if, $w\not\in L(\cm)$.
By Corollary \ref{corollary_cap},
	there exists some error-free computation of $\cm$ from $(s_I,\varepsilon)$ to $(s_F,x)$ for some $x\in M^*$ if, and only if, $L(\A)\cap L(\cm)\neq \emptyset$. 
	The latter is equivalent to $L(\A)\not\subseteq L(\B)$. 
	Hence, the language inclusion problem is undecidable.
		
	\subsection{Undecidability of the Model Checking Problem for \mtl}	
	The proof idea of Theorem \ref{theorem_main_lang_inc} can be used to show the undecidability of the following model checking problem: given a timed visibly one-counter net $\A$, and an \mtl formula $\varphi$, does every $w\in L(\A)$ satisfy  $\varphi$?
Recall that this problem is decidable for the class of timed automata~\cite{DBLP:conf/lics/OuaknineW05}.
We prove that adding a visibly counter without zero test already makes the problem undecidable. 
But first, let us recall the syntax and semantics of MTL. 

Let $\Sigma$ be a  finite alphabet.
	The set of \mtl formulae is built up from $\Sigma$ 
        by Boolean connectives and a time constraining version of the {\em until}
modality:
$$\varphi \ndefeq \true \sep a \sep \neg\varphi \sep \varphi_1\wedge\varphi_2 \sep
        \varphi_1\U_{I}\varphi_2  $$
        where $a\in\Sigma$ and $I \subseteq \RP$ is an open, closed, or half-open interval with endpoints in $\N\cup\{\infty\}$.
%If $I=\RP$, then we may omit the annotation $I$ on $\U_I$.

We interpret \mtl formulae in the \emph{pointwise semantics}, \ie, over finite timed words over $\Sigma$. 
Let $w=(a_1,t_1)(a_2,t_2)\dots(a_n,t_n)$ be a timed word, and let  $i\in\{1,\dots,n\}$. 
We define the {\em satisfaction relation for \mtl}\hspace{-1.4mm}, denoted by $\models$,
inductively as follows:
\begin{flalign*}
	%(w,i) \models \true &  \hspace{2mm}\Leftrightarrow\hspace{2mm} a_i=a \\
	(w,i) \models a &\hspace{2mm}\Leftrightarrow\hspace{2mm} a_i=a\\
	 (w,i) \models
\neg\varphi & \hspace{2mm}\Leftrightarrow\hspace{2mm} (w,i)\not\models \varphi, \\
	 (w,i) \models \varphi_1\wedge\varphi_2 & \hspace{2mm}\Leftrightarrow\hspace{2mm}  
	(w,i)\models \varphi_1 \textrm{ and } (w,i)\models \varphi_2,\\
	 (w,i)\models \varphi_1\U_{I}\varphi_2 & \hspace{2mm}\Leftrightarrow\hspace{2mm} \exists j.
	 i<j\leq n: (w,j)\models\varphi_2 \textrm{ and } t_j - t_i \in I, \\
& \hspace{9.5mm}\textrm{and } \forall k.i< k<j:(w,k)\models\varphi_1.
\end{flalign*}
%Here, $|w|$ denotes the length of the timed word $w$. 
We say that a timed word $w\in\TW$ satisfies an \mtl formula $\varphi$, written $w\models\varphi$, if $(w,1)\models \varphi$.

Note that \mtl only allows to express restrictions on time, and it does not allow for any restrictions on the values of the counters. 
In fact, it is proved that as soon as we add to \mtl the capability for expressing restrictions on the values of a counter that can be incremented and decremented, model checking is undecidable~\cite{DBLP:conf/lata/Quaas13}.
The proof of the following theorem is based on the fact that MTL can encode computations of channel machines with insertion errors~\cite{DBLP:conf/lics/OuaknineW05}.
\begin{theorem}
	\label{theorem_model_checking}
	The model checking problem for timed visibly one-counter nets and \mtl is undecidable, even if the timed visibly one-counter net does not use any clocks and is deterministic. 
\end{theorem}
\begin{proof}
	The definition of an \mtl formula $\varphi$ such that $L(\varphi)=L(\cm)$ is straightforward, see, \eg, ~\cite{DBLP:conf/lics/OuaknineW05}.
	By Corollary \ref{corollary_cap}, 
	there exists some error-free computation of $\cm$ from $(s_I,\varepsilon)$ to $(s_F,x)$ for some $x\in M^*$ if, and only if, $L(\A)\cap L(\varphi)\neq\emptyset$. 
	The latter, however, is equivalent to saying that there exists some timed word $w\in L(\A)$ such that $w\not\models\varphi$. 
	Hence the model checking problem is undecidable. 
\end{proof}

We would like to remark that the proof of Theorem \ref{theorem_model_checking} shares some similarities with the proof of the undecidability of model checking one-counter machines (\ie, one-counter automata without input alphabet) and Freeze LTL with one register (LTL$^\downarrow_1$, for short)~\cite{DLS-tcs10}. 
In~\cite{DBLP:journals/tocl/DemriL09},
it is proved that 
LTL$^\downarrow_1$ can encode computations of \emph{counter automata with incrementing errors}. 
Similar to the situation for \mtl and channel machines, $\mathsf{LTL}^\downarrow_1$ can however not encode \emph{error-free} computations of counter automata. 
In~\cite{DLS-tcs10}, a one-counter machine is used to repair this incapability, resulting in the undecidability of the model checking problem. 
The one-counter machine in~\cite{DLS-tcs10} does not use zero tests; however, we point out that in contrast to our visibly timed one-counter net the one-counter machine in~\cite{DLS-tcs10} is \emph{non-deterministic}. Indeed, model checking \emph{deterministic} one-counter machines and LTL$^\downarrow_1$ is decidable~\cite{DLS-tcs10}.

We further remark that using a similar proof, we can show that the model checking problem for \emph{parametric timed automata} and MTL is undecidable, even if the automaton only uses one parametric clock, one parameter and is deterministic~\cite{DBLP:journals/corr/Quaas14a}.

\subsection{Energy Problems on Timed Automata with Discrete Weights}
Next  we will consider an interesting extension of \emph{lower-bound energy problems on weighted timed automata}, introduced in~\cite{DBLP:conf/formats/BouyerFLMS08}, which gained  attention in the last years, see, \eg,~\cite{DBLP:journals/pe/BouyerLM14,DBLP:conf/lata/Quaas11,DBLP:conf/hybrid/BouyerFLM10}.
In lower-bound energy problems, 
one is interested whether in a given automaton with some weight variable whose value can be increased and decreased, there exists a successful run in which all accumulated weight values are never below zero.  
Similar problems have also been considered for untimed settings, \eg,~\cite{DBLP:conf/birthday/JuhlLR13,DBLP:conf/atva/EsikFLQ13,DBLP:conf/ictac/FahrenbergJLS11,DBLP:conf/icalp/BrazdilJK10}.

A \emph{timed automaton with discrete weights} (dWTA, for short)
is syntactically the same as a timed one-counter net.
	In the semantical graph induced by a dWTA, however, we allow the value of the counter (or, the \emph{weight variable}) to become negative.
	Hence the value of the weight variable does not influence the behaviour of the dWTA, 
	because, different to timed one-counter nets, transitions that result in negative values are not blocked. 
	We remark that for the  simple reasons that (1) the value of the weight variable does not influence the behaviour of dWTA, and (2) MTL does not restrict the values of the weight variable, the model checking problem for dWTA and MTL is decidable, using the same algorithm as for timed automata~\cite{DBLP:conf/lics/OuaknineW05}. 
	We define the \emph{model checking energy problem} for dWTA and \mtl as follows: given a dWTA $\A$ and an \mtl formula $\varphi$, does there exist some accepting run $\rho$ of $\A$ such that the value of the weight variable is always non-negative, and the timed word $w$ associated with $\rho$ satisfies $\varphi$?	
	For the special case $\varphi=\true$, 
	the problem is decidable in polynomial time for one-clock dWTA~\cite{DBLP:conf/formats/BouyerFLMS08}.
	\begin{theorem}
		\label{theorem_energy}
		The model checking energy problem for dWTA and \mtl is undecidable, even if the dWTA uses no clocks. 
	\end{theorem}
	\begin{proof}		
		For the proof, we reduce the  model checking problem for timed one-counter nets and \mtl to the energy problem. 
		Note that timed one-counter nets are a generalization of timed visibly one-counter nets, and thus by Theorem \ref{theorem_model_checking} the model checking problem is undecidable. 
		Let $\A$ be a timed one-counter net, and let $\varphi$ be an \mtl formula. 
		Define $\A'$ to be the dWTA that is syntactically equal to $\A$.
		One can easily prove that $(\A,\varphi)$ is a negative instance of the model checking problem if, and only if, $(\A',\neg\varphi)$ is a positive instance of the energy problem. 
	\end{proof}

	\section{Decidability Result}
	In the preceding section, we showed that one cannot automatically verify  timed automata extended with unbounded discrete data structures against real-time specifications expressed by timed automata or MTL-formulas. 
	In this section, 
	we prove that in contrast to this, we can use one-clock timed counter nets as specification for model checking timed automata: the language inclusion problem $L(\A)\subseteq L(\B)$ is decidable with non-primitive recursive complexity if $\A$ is a timed automaton and $\B$ is a one-clock timed counter net (Theorem \ref{theorem_main_dec}).
	We will first give the formal proof of Theorem \ref{theorem_main_dec}. 
	After that, we will argue that this result extends known facilities for the verification of timed automata. 
	
	\subsection{Proof of Theorem \ref{theorem_main_dec}}
	\begin{proof}
		For the case that $\A$ is a timed automaton and $\B$ is a one-clock timed automaton,
		the language inclusion problem $L(\A)\subseteq L(\B)$ is decidable~\cite{DBLP:conf/lics/OuaknineW04} with non-primitive recursive complexity~\cite{DBLP:journals/fuin/AbdullaDOQW08}. 
		The lower bound hence follows, and the decidability proof is an adaptation of the decidability proof in~\cite{DBLP:conf/lics/OuaknineW04}.
		%We present the idea, 
		%but do not give the details of the proofs, especially if they are similar to the proof in~\cite{DBLP:conf/lics/OuaknineW04}.

		The proof is based on the theory of well-quasi-orders, and we start with defining some useful notions.
	
	Let $A,B$ be two sets, and let $\preceq$ be a binary relation on $A$.
	Then $\preceq$ is a \emph{quasi-order} on $A$ if $\preceq$ is reflexive and transitive. 
	$\preceq$ is a \emph{well-quasi-order} on $A$ if it is a quasi-order and for every infinite sequence $a_1,a_2,a_3,\dots$ in $A$ there exist indices $i<j$ such that $a_i\preceq a_j$. 
	A standard example for a well-quasi-order is the pointwise order $\leq^k$ on the set $\N^k$ of vectors of $k$ natural numbers (Dickson's Lemma, \cite{Dickson1913}). 
	
	Let $\preceq$ be a quasi-order on $A$, 
	and let $\sqsubseteq$ be a quasi-order on $B$. 
	We define the \emph{product} of $\preceq$ and $\sqsubseteq$ on $(A\times B)$ by $(a,b)\leq (a',b')$, if and only if, $a\preceq a'$ and $b \sqsubseteq b'$. 	
	We define the \emph{monotone domination order} $\preceq^*$ on $A^*$ by $a_1 a_2 \dots a_m \preceq^* a'_1 a'_2\dots a'_n$ if and only if there exists a strictly increasing function ${f:\{1,\dots,m\}\to\{1,\dots,n\}}$ such that for all $i\in\{1,\dots,m\}$ we have $a_i\preceq a'_{f(i)}$. 
	We define the \emph{subset order} $\preceq^\mathcal{P}$ on the set $\mathcal{P}(A)$ of finite subsets of $A$ by
	$A_1\preceq^\mathcal{P} A_2$ if and only if there is an injective mapping $f:A_1\to A_2$ such that for all $a\in A_1$ we have $a\preceq f(a)$. 
	
	%$\{a_1,\dots,a_m\}\preceq^\mathcal{P}\{a'_1,\dots,a'_n\}$ if there exists a strictly increasing function $f:\{1,\dots,m\}\to\{1,\dots,n\}$ such that for all $i\in\{1,\dots,m\}$ we have $a_i\preceq a'_{f(i)}$. 
	
	\begin{lemma}[Higman's Lemma~\cite{Higman1952}]
		\label{lemma_higman}
		\begin{enumerate}
		\item If $\preceq$ and $\sqsubseteq$ are well-quasi-orders on $A$ and $B$, respectively, 
			then the product of $\preceq$ and $\sqsubseteq$ is a well-quasi-order on $(A\times B)$. 
	\item  If $\preceq$ is a well-quasi-order on $A$,
		then the monotone domination order $\preceq^*$ is a well-quasi-order on $A^*$. 
	\item  If $\preceq$ is a well-quasi-order on $A$,
		then the subset order $\preceq^\mathcal{P}$ is a well-quasi-order on $\mathcal{P}(A)$. 
	\end{enumerate}
	\end{lemma}

		Let $\A=(\Sigma^\A,\Gamma^\A,\locs^\A,\locs_0^\A,\locs_f^\A,\clocks,\edges^\A)$ be a timed automaton with clock variables $\clocks=\{x_1,\dots,x_n\}$, 
		and let $\B=(\Sigma^\B,\mathfrak{n},\locs^\B,\locs_0^\B,\locs_f^\B,\{x\},\edges^\B)$ be a timed counter net of dimension $\mathfrak{n}$ with a single clock variable $x$.
		Without loss of generality, we may assume that $\Sigma^\A=\Sigma^\B$. 
		%Also note that $\Gamma^\A=\emptyset$ and that the size of the stack alphabet $\Gamma^\B$ is equal to one.
		
		Note that a state $(\loc,\nu)$ of $\A$ is an element in $\locs^\A\times(\RP)^{\clocks}$, 
		and a state $(\loc,v,\vec{u})$ of $\B$ is an element in $\locs^\B\times\RP\times\N^{\mathfrak{n}}$.		
		A \emph{joint configuration of $\A$ and $\B$} is a
		pair $(q,\gamma)$, where $q$ is a state of $\A$, and $\gamma$ is a set of states of $\B$.
		We use $\gs^{\A\B}$ to denote the set of all joint configurations of $\A$ and $\B$. 		
		We say that a joint configuration $(q,\gamma)$ is \emph{initial} if $q\in(\locs_0^\A\times \{0\}^\clocks)$ and $\gamma=\{(\loc,0,\vec{0}) \mid \loc\in\locs^\B_0\}$ (with 
		$\vec{0}$ we denote the vector of dimension $\mathfrak{n}$ containing only $0$).
		We say that a joint configuration $(q,\gamma)$ is \emph{bad} if $q=(\loc,\nu)$ for some $\loc\in\locs_f^\A$, and for all states $(\loc',v',\vec{u}')\in\gamma$ we have $\loc'\not\in\locs_f^\B$. 
		The \emph{joint behaviour of $\A$ and $\B$} is defined as follows: 
		For a state $(\loc,\nu)$ of $\A$, $\delta\in\RP$ and $a\in\Sigma$, 
		we define $\succ^\A((\loc,\nu),\delta,a)=\{(\loc',\nu')\mid \langle(\loc,\nu),\delta,a, (\loc',\nu')\rangle\hspace{-1.4mm}\in\transA\}$.
		%For a state $(\loc,v,u)$ of $\B$, we define $\succ^\B((\loc,v,u),\delta,a)$ analogously, 
		% we do not need this???!!
		%and lift the definition to sets $\gamma$ of states of $\B$ as follows:  
		For a set $\gamma$ of states of $\B$, we define 
		$\succ^\B(\gamma,\delta,a)=\{(\loc',v',\vec{u}')\mid \exists (\loc,v,\vec{u})\in\gamma. \langle(\loc,v,\vec{u}),\delta, a, (\loc',v',\vec{u}')\rangle\hspace{-1.4mm}\in\Rightarrow_\B\}$. Note that $\succ^\B(\gamma,\delta,a)$ is a set of states of $\B$, and it may be empty. 
		Finally, we define the %$\Sigma$-labelled 
		transition relation $\transAB$ on $\gs^{\A\B}\times\Sigma\times\gs^{\A\B}$ by  $\langle(q,\gamma),a,(q',\gamma')\rangle\hspace{-1.4mm}\in\transAB$ if there exists some $\delta\in\RP$ such that $q'\in\succ^\A(q,\delta,a)$ and  $\gamma'=\succ^\B(\gamma,\delta,a)$.

		Next, we encode joint configurations of $\A$ and $\B$ by  finite untimed words over the set $\Lambda$ of finite subsets of $(\locs^\A \times \clocks\times \reg\times\vec{0}) \cup \{\locs^\B \times \{x\}\times \reg \times \N^{\mathfrak{n}})$.
		Here, 	
		$\reg\defeq \{0,1,\dots,\cmax\}\cup\{\top\}$, where $\cmax$ is an integer greater than the maximal constant occurring in clock constraints in both $\A$ and $\B$, and $\top$ is a symbol representing all values greater than $\cmax$. 
		Let $C=\left((\loc,\nu),\{(\loc_1,v_1,\vec{u}_1),\dots,(\loc_m,v_m,\vec{u}_m)\}\right)$
		be a joint configuration.		
		To simplify the definition, we write $C$ as a set 
		$$\{(\loc,x_1,\nu(x_1),\vec{0}),\dots,(\loc,x_n,\nu(x_n),\vec{0}),(\loc_1,x,v_1,\vec{u}_1),\dots,(\loc_m,x,v_m,\vec{u}_m)\}.$$
		Partition $C$ into a sequence of subsets 
		$C_0,C_1,\dots,C_\rho,C_\top$,
		such that $C_\top = \{(k,y,\eta,\vec{\mu})\in C\mid \eta>\cmax\}$, and
		if $i,j\neq\top$,
		then for all $(k,y,\eta,\vec{\mu})\in C_i$, $(k',y',\eta',\vec{\mu'})\in C_j$, 
		we have $\frac(\eta)<\frac(\eta')$ if, and only if, $i<j$, and 
		$\frac(\eta)=\frac(\eta')$ if, and only if, $i=j$. Here, $\frac(r)$ denotes the fractional part of a real number $r$. 		
		In this way,   $(k,y,\eta,\vec{\mu})$ and $(k',y',\eta',\vec{\mu'})$ are in the same subset $C_i$ if, and only if, $\eta$ and $\eta'$ are both smaller than or equal to $\cmax$ and have the same fractional part. 
		In addition, we require that 
		$(k,y,\eta,\vec{\mu})\in C_0$ if, and only if, the fractional part of $\eta$ is zero, and
		$C_i\neq\emptyset$ for all $i\in\{1,\dots,\rho\}$. 
		We define the encoding $\enc(C)$ of $C$ to be  the finite word 
		$\reg(C_0)\reg(C_1)\dots \reg(C_\rho)\reg(C_\top)$, 
		where $\reg(C_i)=\{(k,y,\reg(\eta),\vec{\mu})\mid (k,y,\eta,\vec{\mu})\in C_i\}$ with $\reg(\eta)=\int(\eta)$ if $\eta\leq\cmax$, and $\reg(\eta)=\top$ otherwise ($\int(r)$ denotes the integer part of a real number $r$).

		We define a transition relation $\to$ on the set of encodings of joint configurations and $\Sigma$ as follows: $\langle w, a, w'\rangle\hspace{-1.2mm}\in\to$ if there exists $C\in \enc^{-1}(w)$ and $C'\in \enc^{-1}(w')$ such that $\langle C,a, C'\rangle\hspace{-1.4mm}\in\transAB$. 
		We further define the equivalence relation $\sim$ by $C\sim C'$ if, and only if, $\enc(C)=\enc(C')$.
		\begin{example}
			Let $\A$ be a timed automaton with a single clock $y$, and let $\B$ be a timed one-counter net of dimension $2$ and with a single clock $x$. Assume $\cmax=2$. 
			Let $$C_2=\langle (\loc,1.3),\{(\loc_2,0.7,(1,1)),(\loc_1,1.0,(1,1)),(\loc_3,0.5,(0,1)),(\loc_1,2.2,(0,0))\}\rangle$$ be a joint configuration of $\A$ and $\B$. 
			The encoding of $C_2$ equals 
			$$\enc(C_2)=
			\{(\loc_1,x,1,(1,1))\}\{(\loc,y,1,(0,0))\}\{(\loc_3,x,0,(0,1))\}\{(\loc_2,x,0,(1,1))\}\{(\loc_1,x,\top,(0,0))\}$$
			The joint configuration 
			$$C_3=\langle(\loc,1.1),\{(\loc_2,0.9,(1,1)),(\loc_1,1.0,(1,1)),(\loc_3,0.2,(0,1)),(\loc_1,9.2,(0,0))\}\rangle$$
			has the same encoding, and thus $C_2\sim C_3$. Note that for $C_2\sim C_3$ to hold, the configurations must agree on the counter values. 
		\end{example}
		In the next lemma, we prove that $\sim$ is a time-abstract bisimulation over joint configurations. 
		The proof can be done like the proof of Prop. 11 in~\cite{DBLP:conf/lics/OuaknineW04}. 
		\begin{lemma}
			\label{lemma_bisim}
			For all joint configurations $C_1, C_2$, and $a\in\Sigma$, if $C_1\sim C_2$, then
			\begin{itemize}
			\item for all $C'_1$ such that $\langle C_1, a, C'_1\rangle\in\transAB$, there exists $C'_2$ such that $\langle C_2,a,C'_2\rangle\in\transAB$ and $C'_1\sim C'_2$,
			\item for all $C'_2$ such that $\langle C_2, a, C'_2\rangle\in\transAB$, there exists $C'_1$ such that $\langle C_1,a,C'_1\rangle\in\transAB$ and $C'_1\sim C'_2$.
			\end{itemize}
		\end{lemma}
		%NOTE: if H(C)=H(C') then the weights are the same. If the weights are not the same, \sim is not a bisimulation!
		\noindent Next, we define a quasi-order $\sqsubseteq$ on the set of encodings of joint configurations and prove that $\sqsubseteq$ is a well-quasi-order. 
		First, define $\preceq$ on 
		$(\locs^\A \times \clocks\times \reg\times\vec{0}) \cup (\locs^\B \times \{x\}\times \reg \times \N^{\mathfrak{n}})$ 
		by $(k,y,\eta,\mu)\preceq (k',y',\eta',\mu')$ if, and only if, $k=k'$, $y=y'$, $\eta=\eta'$, and $\mu\leq^{\mathfrak{n}}\mu'$. 
		By Lemma \ref{lemma_higman}.1 and the fact that $=$ on the finite set $(\locs^\A \times \clocks\times \reg) \cup (\locs^\B \times \{x\}\times \reg)$ and $\leq^{\mathfrak{n}}$ on $\N^{\mathfrak{n}}$ are well-quasi-orders, 
		$\preceq$ is a well-quasi-order, too.
		By Lemma \ref{lemma_higman}.3, the subset order $\preceq^\mathcal{P}$ is a well-quasi-order. 
		Finally, we define $\sqsubseteq$ to be the monotone domination order on $\preceq^\mathcal{P}$, and then by Lemma \ref{lemma_higman}.2, 
		$\sqsubseteq$ is a well-quasi-order. 
		\begin{example}
			%Let $C_1 = \langle(\loc,1.4),\{(\loc_2,0.5,(1,0)),(\loc_1,1.0,(1,1))\}\rangle$, and thus $\enc(C_1)=\{(\loc_1,x,1,(1,1))\}\{(\loc,y,1,(0,0))\}\{(\loc_2,x,0,(1,0))\}\rangle$. 
			Let $w_1 = \{(\loc_1,x,1,(1,0))\}\{(\loc,y,1,(0,0))\}\{(\loc_2,x,0,(1,1))\}\rangle$.
			Then $w_1 \sqsubseteq w_2$, where $w_2=\enc(C_2)$ from the previous example. Note that the counter values in $w_1$ may be smaller than the associated counter values in $w_2$, as it is here the case with the counter values for $\loc_1$. 
		\end{example}		
		The next lemma states that $\to$ is downward-compatible with respect to $\sqsubseteq$. 
		\begin{lemma}
			\label{lemma_dc}
			If $w_1\sqsubseteq w_2$ and $\langle w_2,a, w'_2\rangle\in\to$, then there exists $w'_1$ such that $w'_1\sqsubseteq w'_2$ and $\langle w_1,a, w'_1\rangle\hspace{-1.1mm}\in\to$.
		\end{lemma}
		\begin{proof}
			The proof is similar to the proof of Lemma 15 in~\cite{DBLP:conf/lics/OuaknineW04}.
			
			Let $w_1, w_2, w_2'$ be such that $w_1\sqsubseteq w_2$ and $\langle w_2,a,w'_2\rangle\in\to$.
			Further let 
			$C_1=((\loc_1,\nu_1),\gamma_1)\in \enc^{-1}(w_1)$, 
			$C_2=((\loc_2,\nu_2),\gamma_2)\in\enc^{-1}(w_2)$, 
			and $C'_2=((\loc'_2,\nu'_2),\gamma'_2)\in\enc^{-1}(w'_2)$ be such that \newline$\langle C_2,a, C'_2\rangle\in\transAB$. 
			This implies that 
			there exists $\delta\in\RP$ with
			$(\loc'_2,\nu'_2)\in\succ^\A((\loc_2,\nu_2),\delta,a)$ and
			$\gamma'_2=\succ^\B(\gamma_2,\delta,a)$.

			%From $w_1\sqsubseteq w_2$, we can conclude that $(\loc_1,\nu nee das wird zu kompliziert mit $\nu_1$ und $\nu_2$....
			
			Since $w_1\sqsubseteq w_2$, 
			we know that there exists a set $\gamma_{minus}$ of states in  $\B$  such that
			$((\loc_2,\nu_2),\gamma_{minus})\sim ((\loc_1,\nu_1),\gamma_1)$ ($\star$): $\gamma_{minus}$ is obtained from choosing a suitable set $\gamma''_2\subseteq \gamma_2$
			such that the encoding of $((\loc_2,\nu_2),\gamma''_2)$ differs from the encoding of $((\loc_1,\nu_1),\gamma_1)$ only in that the vectors representing the counter values occurring in $\gamma''_2$ may be greater (with respect to $\leq^{\mathfrak{n}}$) than the corresponding vectors in $\gamma_1$.  
			$\gamma_{minus}$  is then the result from subtracting suitable values from the vectors in $\gamma''_2$ so that the encoding is equal. 
			%(We denote this operation by...)

			%Since $w_1\sqsubseteq w_2$, 
			%we know that there exists 
			%$\eta_2 \subseteq \gamma_2$
			%and a sequence (list) of vectors $v=\{v_1,\dots,v_p\}\subseteq \N_{\le 0}^\mathfrak{n}$ such that $p=|\gamma_1|$ (to define!) and $\eta_2-v=\eta$ (to define!) and 
			%$((\loc_1,\nu_1),\gamma_1)\sim ((\loc_2,\nu_2),\eta)$. 

			Now let $\gamma'_{minus} = \succ^\B(\gamma_{minus},\delta,a)$. 
			Then $\langle (\loc_2,\nu_2),\gamma_{minus}),a,(\loc'_2,\nu'_2),\gamma'_{minus})\rangle\in\transAB$ ($\star\star$). 
			Note that we can add suitable values to the vectors in $\gamma'_{minus}$ to obtain $\gamma'_{minus}\subseteq \gamma'_2$ ($\star\star\star$).

			%Now I would like to present the fact that $\gamma'_{minus}$ plus the values that we subtracted will result in a set of states in $\B$ that is a subset of $\gamma'_2$. 

			From $(\star)$ and $(\star\star)$ it follows by Lemma \ref{lemma_bisim} that there exists
			$C'_1=((\loc'_1,\nu'_1),\gamma'_1)$ such that 
			$\langle C_1,a,C'_1\rangle\in\transAB$ and 
			$C'_1 \sim ((\loc'_2,\nu'_2),\gamma'_{minus})$. 
			From the former it follows that
			$\langle w_1, a, w'_1\rangle\in\to$, where $w'_1=\enc(C'_1)$. 
			From the second and $(\star\star\star)$ it follows that  $w'_1\sqsubseteq w'_2$.
		\end{proof}
		Intuitively, 
		comparing the situation for timed counter nets with the situation for pure timed automata like in Lemma 15 in~\cite{DBLP:conf/lics/OuaknineW04}, 
		there may now be transitions that can be executed from configurations encoded by $w_2$, but that are blocked from configurations encoded by $w_1$ due to the fact that counter values are too small. This, however, does not cause any trouble, because it results in smaller sets of successor configurations, and thus leading to $w'_1\sqsubseteq w'_2$. 
		\begin{remark}
		Note that Lemma  \ref{lemma_dc} does not hold if the counters in $\B$ can be tested for zero. 
		For instance, consider $w_1$ and $w_2$ from the previous example. Assume that in $\B$ there are no edges with source location $\loc_2$ and $\loc_3$, and the only suitable $a$-labelled edge with source location $\loc_1$ 
		does a zero test on the second counter and leaves all other components unchanged. 
		This yields  
		$\succ^\B(\gamma_2,0,a)=\emptyset$ and $\succ^\B(\gamma_1,0,a)=\{(\loc_1,x,1,(1,0))\}$, 
		where $\gamma_i$ is the set of states of $\B$ in some configuration $C_i$ with $C_i\in\enc^{-1}(w_i)$ for $i=1,2$.
		Assume there is an $a$-labelled edge in $\A$ with source location $\loc$ and leaving all components unchanged. 
		Then we have
		$\langle w_2, a, w'_2\rangle\in\to$ with  $w'_2=\emptyset \{(\loc,y,1,(0,0)\}\emptyset$, 
		$\langle w_1, a, w'_1\rangle\in\to$ with 
		$w'_1 = \{(\loc_1,x,1,(1,0))\}\{(\loc,y,1,(0,0))\}\emptyset$. 
		Note that $w'_1\sqsubseteq w'_2$ does \emph{not} hold. 		
		Indeed, the universality problem of (even untimed) one-counter automata is undecidable~\cite{DBLP:journals/jacm/Greibach69,DBLP:journals/mst/Ibarra79}. 
		In Section 6, we prove that this is also the case for timed visibly pushdown one-counter automata. 
		This gives us the precise decidability border for the universality problem. 
	\end{remark}
		%\textcolor{red}{
		%\begin{proof}
		%	If an edge  decrementing the counter cannot be executed from the smaller word $w_1$, then we obtain a transition $\Rightarrow$ to the empty set, which is smaller than $w_2$ (which is a valid transition). 
		%	If something like this would be in the $\A$ component, then we would not get a transition $\Rightarrow$ at all.
		%\end{proof}}
		Finally, we describe the algorithm to decide $L(\A)\subseteq L(\B)$. 
		Like in~\cite{DBLP:conf/lics/OuaknineW04}, 
		we solve the language inclusion problem by solving the following reachability problem: in the implicit graph of the encoding of joint configurations and the transition relation $\to$,
		is there a path from the encoding of one of the finitely many  initial joint configuration to the encoding of a bad joint configuration?
		Note that we have $L(\A)\subseteq L(\B)$ if, and only if, there is no such path. 
		For solving the reachability problem, we compute the unfolding of the graph, 
		starting the computation with the encoding of an initial joint configuration.
		If for the current node labelled by $w$, there is along the branch already a node labelled with $w'$ and $w'\sqsubseteq w$, 
		then by Lemma \ref{lemma_dc} we can prune the tree after the current node: 
		Assume that from $w$ we can reach a word $w_1$ that represents a bad configuration,
		then by Lemma \ref{lemma_dc} we can reach a word $w'_1$ from $w'$ such that $w'_1\sqsubseteq w_1$, and hence, $w'_1$ is representing a bad configuration, too. 		
		By the facts that the unfolding is finitely branching, $\sqsubseteq$ is a well-quasi-order and by K\"onig's Lemma, we know that the computation will finally terminate. 
	\end{proof}

	\subsection{On the Expressiveness of Timed Counter Nets}
	Theorem \ref{theorem_main_dec} generalizes a result by Ouaknine and Worrell on the decidability of the language inclusion problem $L(\A)\subseteq L(\B)$ for $\B$ being a one-clock timed automaton without counters~\cite{DBLP:conf/lics/OuaknineW04}.
	Clearly, timed counter nets are more expressive than timed automata; for instance, the timed language accepted by the timed one-counter net on the left hand side of Figure \ref{tocn_L_2} cannot be accepted by any timed automaton, because its projection on $\Sigma^+$ equals $\{a^nb^m\mid n\geq m, n\geq 1\}$, \ie, a non-regular language. 
		
		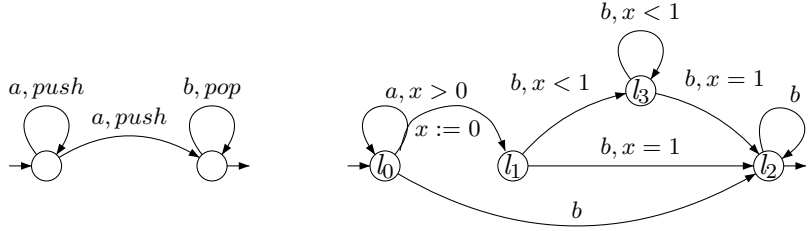
\begin{figure}
\begin{center}
		\begin{picture}(100,25)(0,-25)
%\put(0,-22){\framebox(110,22){}}
\node[NLangle=0.0,Nmarks=i,flength=3,ilength=3,Nw=4.0,Nh=4.0,Nmr=2.0](n0)(5.0,-16.0){}
\node[NLangle=0.0,Nmarks=f,flength=3,Nw=4.0,Nh=4.0,Nmr=2.0](n1)(27.0,-16){}
\drawedge[curvedepth=4.0](n0,n1){\footnotesize{$a,push$}}
\drawloop[loopdiam=6](n0){\footnotesize{$a,push$}}
\drawloop[loopdiam=6](n1){\footnotesize{$b,pop$}}

\node[NLangle=0.0,Nmarks=i,ilength=3,Nw=4.0,Nh=4.0,Nmr=2.0](m0)(50.0,-16.0){$\loc_0$}
\node[NLangle=0.0,Nw=4.0,Nh=4.0,Nmr=2.0](m1)(67.0,-16){$\loc_1$}
\node[NLangle=0.0,Nw=4.0,Nh=4.0,Nmr=2.0](m3)(84,-6){$\loc_3$}
\node[NLangle=0.0,Nmarks=f,flength=3,Nw=4.0,Nh=4.0,Nmr=2.0](m2)(101,-16){$\loc_2$}

\drawedge[curvedepth=2.0](m1,m3){\footnotesize{$b,x<1$}}
\drawloop[loopdiam=6](m3){\footnotesize{$b,x<1$}}
\drawedge[curvedepth=2.0](m3,m2){\footnotesize{$b,x=1\ \ $}}

\drawedge(m1,m2){\footnotesize{$b,x=1$}}
\drawedge[curvedepth=-8.0](m0,m2){\footnotesize{$b$}}
\drawloop[loopdiam=6,loopangle=70](m2){\footnotesize{$b$}}
%\put(83,-7){\footnotesize{$x=1$}}

\drawloop[loopdiam=6,loopCW=n](m0){}
%\drawloop[loopdiam=6,loopangle=270](m0){b}
%\drawloop[loopdiam=6,loopangle=270](m1){a}

\put(50,-7){\footnotesize{$a,x>0$}}
\put(54,-12){\footnotesize{$x:=0$}}
\put(52,-18){
    \unitlength=4mm
   
    \drawcurve[AHnb=0,AHnb=1](0,1)(0.4,2)(2,2.5)(3,2)(3.5,1)
  }
\end{picture}
\caption{From left to right: A timed one-counter net and an alternating one-clock timed automaton}
\label{tocn_L_2}
\end{center}
\end{figure}
However, the result in~\cite{DBLP:conf/lics/OuaknineW04} 
was also generalized to another extension of timed automata called \emph{alternating} one-clock timed automata~\cite{DBLP:conf/lics/OuaknineW05,DBLP:journals/tocl/LasotaW08}.
An alternating one-clock timed automaton allows for two modes of branching, namely existential branching and universal branching, represented by disjunction and conjunction, respectively. 
For example, the alternating one-clock timed automaton on the right hand side of Figure \ref{tocn_L_2} has a universal branching transition in $\loc_0$ for the input letter $a$, formally expressed by $\delta(\loc_0,a)=x>0\wedge(\loc_0 \wedge x.\loc_1)$; and it has an existential branching transition in $\loc_1$ for the input letter $b$, formally $\delta(\loc_1,b)=(x<1\wedge \loc_3)\vee(x=1\wedge\loc_2)$  (see~\cite{DBLP:conf/lics/OuaknineW05} for more details). 
This alternating one-clock timed automaton accepts the timed language consisting of a sequence of $a$'s followed by a sequence of $b$'s such that the time sequence belonging to the $a$-sequence is strictly monotonic, and every $a$ is followed by some $b$ after exactly one time unit. Note that the projection on $\Sigma^+$ thus equals $\{a^n b^m \mid m\geq n, m\geq 1\}$.  

We prove that timed one-counter nets with one clock and alternating one-clock timed automata are incomparable in expressive power. 
\begin{theorem}
	Timed one-counter one-clock nets and alternating one-clock timed automata are incomparable in expressiveness.
\end{theorem}
\begin{proof}
	On the hand, due to the lack of zero tests, 
	the timed language accepted by the alternating one-clock timed automaton on the right hand side of Figure \ref{tocn_L_2} cannot be accepted by any timed one-counter net. 
	On the other hand, 
	we prove that the timed language $L_{\geq}$ accepted by the timed counter net on the left side of Figure \ref{tocn_L_2} cannot be accepted by any alternating one-clock timed automaton, as we will prove in the following.

	We start with some simple facts about deterministic finite automata. 	
	Let $\B$ be a deterministic finite automaton over the singleton alphabet $\{b\}$ and with a set of states denoted by $Q$. 
	For every $n\in\N$, 
	we define a function $f^\B_n:Q\to Q$ such that $f^\B_n(q)=q'$ means that if $\B$ starts in state $q$ to read the word $b^n$, then $\B$ ends in $q'$. 
	Clearly, there exist natural numbers $k,m \ge 1$ such that $f^\B_{k}=f^\B_{k+m}$. 
	By determinism of $\B$ we further have $f^\B_{k+i}=f^\B_{k+m+i}$ for every $i\ge 1$. 
	
	Now let $\{\B_1,\dots,\B_n\}$ be a finite set of deterministic finite automata. 
	For every $i\in\{1,\dots,n\}$, let $k_{i}$ and $m_{i}$ be such that $f^{\B_i}_{k_i}=f^{\B_i}_{k_i+m_i}$. 
	Set $k=\max\{k_1,\dots,k_n\}$, and let $m$ be the least common multiple of $m_1,\dots,m_n$. 
	One can easily prove that for every $i\in\{1,\dots,n\}$ we have $f^{\B_i}_k = f^{\B_i}_{k+m}$, and, again by determinism, 
	$f^{\B_i}_{k+j} = f^{\B_i}_{k+m+j}$ for every $j\geq 1$.

	Assume by contradiction that $L_{\geq}$ is accepted by an alternating  one-clock timed automaton $\A$.
	Assume $\A$ has $n$ locations.
	Let $B$ be the set of all deterministic finite automata over $\{b\}$ with at most $2^{2^{2n}}$ states. Note that $B$ is finite up to equivalent behaviour. 
	Now choose $k,m \ge 1$ as explained above and such that 
	$f^\B_{k+j}=f^\B_{k+m+j}$ for every $\B\in B$ and $j\ge 0$.

	%Let $k<m$ be such that $f^\mathcal{B}_{k+i}=f^\mathcal{B}_{k+m+i}$ for all $i>0$ and all deterministic finite automata $\B$ with at most $2^{2^{2n}}$ states.

	Define $\delta = {1\over {2k+2m+2}}$. 
	Define $w_1 = (a^{k+m} b^{k+1},\tau)$
	and define $w_2 = (a^{k+m} b^{k+m+1},\tau')$, 	
	where $\tau = t_1 t_2 \dots t_{2k+m+1}$ with $t_i=i \cdot \delta$ for every $i\in\{1,\dots, 2k+m+1\}$, 
	and, similarly, $\tau' = t'_1 t'_2 \dots  t'_{2k+2m+1}$ with $t'_{i}=i\cdot\delta$ for every $i\in\{1,\dots, 2k+2m+1\}$.
	Note that $w_1\in L_\geq$ and $w_2\not\in L_\geq$.
	Further note that $0<t_i<1$ and $0<t'_i<1$ for all $i\in\{1,2k+2m+1\}$.

	We prove that $w_1\in L(\A)$ if, and only if, $w_2\in L(\A)$, \ie, $\A$ cannot distinguish between $w_1$ and $w_2$.
	
	First assume that $w_1\in L(\A)$.	
	Let $\gamma$ be the configuration that $\A$ is in after reading $a^{k+m}$. 
	Clearly, all states $(\loc,x)$ in $\gamma$ satisfy $0\leq x<1$. 
	Let $\rho$ be the run of $\A$  that starts from $\gamma$ on the suffix of $w_1$ that contains the $k+1$ many $b$'s. 
	All clock constraints occurring in transitions of $\rho$ are of the form $x\sim c$ for some $c\in\N$, and by the choice of $t_i$, 
	%all clock constraints are irrelevant for the acceptance of $w$. 
	% all t_i are between 0 and 1, hence what may be relevant is x=0, or x>0.
	% however, the time delays all are >0
	the only  clock constraints that are relevant for the acceptance of $w_1$ are those with $c$ equal to $0$. The satisfaction of constraints of the form $x\sim 0$ may depend on preceding resets of the clock $x$; however, even with clock resets occurring in $\rho$, the clock constraint $x=0$ cannot be satisfied anywhere in $\rho$ because the time delays between the $b$'s are always greater than $0$.
	In other words, $\A$ behaves on the (untimed) word $b^{k+1}$ like an alternating automaton without a clock, but with an additional flag telling whether there was a reset on the clock or not. This, however, is equivalent to the behaviour of a  deterministic finite automaton with $2^{2^{2n}}$ states on the (untimed) word $b^{k+1}$. 
	But then, by the choice of $k$ and $m$, 
	we know that starting from $\gamma$, $\A$ also accepts $b^{k+m+1}$, and thus $w_2\in L(\A)$. 	
	The proof for the other direction is analogous. 
	%Why $2^{2^{2n}}$? 
	%alternating timed automaton with $n$ states, no clock, but a flag that determines whether the clock is reset or not.
	%For every location we have $2n$ possible values where to go: for everx $\loc$: $x.\loc$ or $\loc$. These $2n$ values can be combined by $\wedge$ and $\vee$: $2^{2n}$. Now we consider the set of all subsets of this, i.e., $2^{2^{2n}}$. 	
\end{proof}

		\section{The Universality Problem for Visibly One-Counter Automata}
		We prove that allowing zero tests in a one-clock timed visibly one-counter net results in the undecidability of the universality problem.
		The undecidability of the universality problem for the more general class of one-clock visibly pushdown automata was already stated in Theorem 3 in~\cite{EmmiM06}.
		The proof in~\cite{EmmiM06} is a reduction of the halting problem for two-counter machines. 
		Given a two-counter machine $\tcm$, 
		one can define a timed language $L(\tcm)$ that consists of all timed words encoding a halting computation of $\tcm$. 
		Then a timed visibly pushdown automaton $\A$ is defined that accepts the complement of $L(\tcm)$. 
		Altogether, $L(\A)=T\Sigma^+$ if, and only if, $\tcm$ does not have a halting computation. 
		The definition of $L(\tcm)$ is similar to the definition of $L(\cm)$ in the proof of Theorem \ref{theorem_main_lang_inc}. 
		Recall that in the definition of $L(\cm)$ we did not include a condition that requires every symbol to have a matching symbol two time units \emph{before}, and, as we mentioned, 
		this is the reason for $L(\cm)$ to contain timed words encoding \emph{faulty} computations of $\cm$. 
		However, in the definition of $L(\tcm)$ in~\cite{EmmiM06}, such a ``backward-looking'' condition is used. In the proof in~\cite{EmmiM06}, it is unfortunately not clear how the  one-clock timed visibly pushdown automaton $\A$ can detect violations of this condition\footnote{More detailed, it is not clear how to construct one-clock timed automata $N_{\neg f_r \leftarrow f_c}$ and $N_{\neg g_r\leftarrow g_c}$ mentioned on p. 10 in~\cite{EmmiM06}. Recall that in the proof for undecidability of the universality problem for timed automata with two or more clocks, it is exactly this backward-looking condition that requires \emph{two} clocks~\cite{DBLP:conf/formats/AdamsOW07}.}. 
		
		Here, we give a complete proof for the subclass of timed visibly one-counter automata. 
		Like the proof of Theorem \ref{theorem_main_lang_inc}, the proof is a reduction of the control state reachability problem for channel machines. 
		We however remark that one can similarly use a reduction of the halting problem for two-counter machines.  		
		
		\
	
		\noindent
		{\bf Proof of Theorem \ref{theorem_univ_undec}} 
		Let $\cm=(S,s_I,M,\Delta)$ be a channel machine, and let $s_F\in S$. 
		Define $\Sigma$ in the same way as in the proof of Theorem \ref{theorem_main_lang_inc}.
		For every $n\in\N$, we define a timed language $L_{\mathsf{ef}}(\cm,n)$
		that 
		consists of all timed words over $\Sigma$ that encode %\emph{error-free} 
		computations of $\cm$ from $(s_I,\varepsilon)$ to $(s_F,x)$ for some $x\in M^*$. But in contrast to the proof of Theorem \ref{theorem_main_lang_inc}, 
		$L(\cm,n)$ will only contain encodings of \emph{error-free} computations of $\cm$.

		Formally, $L_{\mathsf{ef}}(\cm,n)$ is defined using the same conditions as the ones for $L(\cm,n)$ in the proof of Theorem \ref{theorem_main_lang_inc} plus an additional condition that requires the number of wildcard symbols in the enoding of the last configuration to be equal to $n$:		
		\begin{enumerate}
		\addtocounter{enumi}{9} 
	\item Between $s_F$ and $\star$, the wildcard symbol $-$ occurs exactly $n$ times. 
		\end{enumerate}	
		Recall that the conditions (1) to (9) guaranteed that the length of consecutive encodings cannot decrease, \ie, the length of every encoding is at least $n+1$. However, insertion errors may occur, leading to an increase of the length of the encoding and all consecutive encodings.  
		But by the new condition (10), we can exclude the occurrence of such insertion errors. 	We thus have for $L_{\mathsf{ef}}(\cm)=\bigcup_{n\geq 1}L_{\mathsf{ef}}(\cm,n)$:
		\begin{lemma}
			\label{lemma_univ_automata}
			There exists some error-free computation of $\cm$ from $(s_I,\varepsilon)$ to $(s_F,x)$ for some $x\in M^*$, if, and only if, $L_{\mathsf{ef}}(\cm)\neq\emptyset$. 
		\end{lemma}
		%\begin{proof}
		%	The direction from right to left is straightforward. 
		%	For the other direction assume $\gamma$ is an error-free computation of $\cm$ from $(s_I,\varepsilon)$ to $(s_F,x)$ for some $x\in M^*$.
		%	Put $n=\max(\gamma)$. 
		%	By Lemma \ref{lemma_comp_to_word},
		%	there exists $w\in L(\cm,n)$ such that $\max(w)=n$. 
		%	A close inspection of conditions 11, 12, 13, 11', 12' and 13' reveals that we can also ensure that condition 14 is satisfied. Hence $w\in L_{\mathsf{ef}}(\cm)$. 
			
		%\end{proof}
		Next, we define a timed visibly one-counter automaton with a single clock such that $L(\A)=T\Sigma^+\backslash L_{\mathsf{ef}}(\cm)$. 
		Hence, by the preceding lemma, $L(\A)\neq T\Sigma^+$ if, and only if, there exists some error-free computation of $\cm$ from $(s_I,\varepsilon)$ to $(s_F,x)$ for some $x\in M^*$.
		
		$\A$ is the union of $\B$ defined in the proof of Theorem \ref{theorem_main_lang_inc} and the visibly one-counter automaton shown in Figure \ref{figure_visibly_one_counter_automaton}. The latter accepts timed words violating the new condition (10): 
		The automaton non-deterministically guesses the maximum number $n$ of occurrences of the symbol $+$. 
		When leaving $\loc_1$, 
		the value of the counter is $n+1$. 
		The final location $\loc_4$, however, 
		can only be reached while reading $-$ or $\star$ if the value of the counter is zero. This means that the encoding of the last configuration contains at least one symbol more than the encoding of the initial configuration. 
	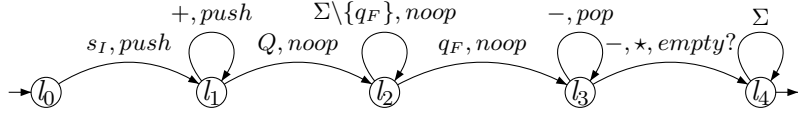
\begin{figure}
\begin{center}
		\begin{picture}(107,14)(0,-14)
%\put(0,-22){\framebox(110,22){}}
\node[NLangle=0.0,Nmarks=i,ilength=3,Nw=4.0,Nh=4.0,Nmr=2.0](n0)(5.0,-11.0){$\loc_0$}
\node[NLangle=0.0,Nw=4.0,Nh=4.0,Nmr=2.0](n1)(27.0,-11){$\loc_1$}
\node[NLangle=0.0,Nw=4.0,Nh=4.0,Nmr=2.0](n2)(50.0,-11){$\loc_2$}
\drawedge[curvedepth=4.0](n0,n1){\footnotesize{$s_I,push$}}
\drawloop[loopdiam=6](n1){\footnotesize{$+,push$}}
\drawedge[curvedepth=4.0](n1,n2){\footnotesize{$Q,noop$}}
\drawloop[loopdiam=6](n2){\footnotesize{$\Sigma\backslash\{q_F\},noop$}}
\node[NLangle=0.0,Nw=4.0,Nh=4.0,Nmr=2.0](n3)(76.0,-11){$\loc_3$}
\node[NLangle=0.0,Nmarks=f,flength=3,Nw=4.0,Nh=4.0,Nmr=2.0](n4)(100.0,-11){$\loc_4$}
\drawedge[curvedepth=4.0](n2,n3){\footnotesize{$q_F,noop$}}
\drawloop[loopdiam=6](n3){\footnotesize{$-, pop$}}
%\drawloop[loopdiam=6,loopangle=270](n3){\footnotesize{$\Sigma \backslash (M\cup \{\star,\#\}), noop $}}
\drawedge[curvedepth=4.0](n3,n4){\footnotesize{$-,\star, empty?$}}
\drawloop[loopdiam=6](n4){\footnotesize{$\Sigma$}}
\end{picture}
\caption{The  timed visibly one-counter automaton  for recognizing timed words violating the additional ``backwards-looking'' condition of $L_{\mathsf{ef}}(\cm)$. }
\label{figure_visibly_one_counter_automaton}
\end{center}
\end{figure}

\section{Conclusion and Open Problems}
The main conclusion of this paper is that even for very weak extensions of timed automata with counters it is impossible to automatically verify whether a given specification is satisfied. 
On the other hand, we may use one-clock timed counter nets as specifications to verify timed automata.
The results on the expressive power of timed counter nets in Sect. 5.2 show that this increases so far known possibilities for the verification of timed automata. 

An interesting problem is to figure out a (decidable) extension of LTL that is capable of expressing properties referring to both time and discrete data structures. 

We remark that all our results hold for automata defined over \emph{finite} timed words. 
We cannot expect the decidability of, \eg, the universality problem for one-clock timed counter nets over infinite timed words, as the same problem is already undecidable for the subclass of one-clock timed automata~\cite{DBLP:journals/fuin/AbdullaDOQW08}.

\subsection*{Acknowledgements}
I would like to thank Michael Emmi and Rupak Majumdar for helpful discussions on their work on timed pushdown automata. 
I further would like to thank James Worrell  very much for pointing me to MTL's capability of encoding faulty computations of channel machines.

%% in general the use of bibtex is encouraged

\end{document}